\documentclass[reqno,oneside]{amsart}
\usepackage{amsthm,amstext,amscd,amsbsy,amsxtra,latexsym}
\usepackage[english]{babel}
\usepackage[latin1]{inputenc}
\usepackage{verbatim} 
\usepackage{breqn}

\usepackage[all]{xy} 
\numberwithin{figure}{section}

\numberwithin{equation}{section}

\newtheorem{thm}{\textbf{Theorem}}
\numberwithin{thm}{section}
\newtheorem{cor}{\textbf{Corollary}}
\numberwithin{cor}{section}
\newtheorem{lem}[thm]{\textbf{Lemma}}
\newtheorem{prop}[thm]{\textbf{Proposition}}
\newtheorem{conj}[thm]{Conjecture}

\theoremstyle{definition}
\newtheorem{defn}[thm]{Definition}
\newtheorem*{rem}{Remark}

\newcommand{\bea}{\begin{eqnarray}} 
\newcommand{\eea}{\end{eqnarray}} 
\newcommand{\be}{\begin{equation}} 
\newcommand{\ee}{\end{equation}} 
\newcommand{\benn}{\begin{equation*}} 
\newcommand{\eenn}{\end{equation*}}

\begin{document}

\title[Quantizing Weierstrass]{Quantizing Weierstrass}

\author[V.~Bouchard, N.K.~Chidambaram and T.~Dauphinee]{Vincent Bouchard, Nitin K. Chidambaram and Tyler Dauphinee}

\address{Department of Mathematical \& Statistical Sciences\\
University of Alberta\\
632 CAB, Edmonton, Alberta T6G 2G1\\
Canada}
\email{vincent.bouchard@ualberta.ca}
\email{chidamba@ualberta.ca}
\email{tdauphin@ualberta.ca}
%\received{November 26, 2014}

\begin{abstract}
We study the connection between the Eynard-Orantin topological recursion and quantum curves for the family of  genus one spectral curves given by the Weierstrass equation. We construct differential operators that annihilate the perturbative and non-perturbative wave-functions. In particular, for the non-perturbative wave-function, we prove, up to order $\hbar^5$, that the differential operator is a quantum curve. As a side result, we obtain an infinite sequence of identities relating $A$-cycle integrals of elliptic functions and quasi-modular forms.
\end{abstract} 

\maketitle

\section{Introduction}

The starting point of the Eynard-Orantin topological recursion \cite{EO,EO2} is a spectral curve. For the purposes of this paper, we can think of a spectral curve as an irreducible algebraic curve $\{ P(x,y) = 0 \} \subset\mathbb{C}^2$. Then, the topological recursion  recursively constructs an infinite sequence of symmetric meromorphic differentials $W_{g,n}$, $g \geq 0$ and $n \geq 1$, on the spectral curve. Depending on the choice of spectral curve, these differentials turn out to be generating functions for many different types of enumerative invariants, such as Gromov-Witten invariants, Hurwitz numbers, knot invariants, etc.\footnote{To be precise, for most of these applications the definition of the spectral curve must be generalized slightly.} (See for instance \cite{BoE,BEMS,BKMP,BM,BSLM,DFM,DLN,DBOSS,DKOSS,DBMNPS,EMS,EO,EO2,EO3,FLZ,FLZ2,FLZ3,GJKS,Ma,MSS}.)

The Eynard-Orantin topological recursion originated in the context of matrix models \cite{CEO,E1,EO,EO2}. But given its rather universal enumerative geometric interpretation, it has now a life of its own. However, it is still interesting to explore its roots, and see whether matrix model theory suggests further connections to \emph{a priori} unrelated mathematical structures. Those may lead to unexpected results in the various geometric contexts.

One such connection that is suggested by matrix models relates the Eynard-Orantin topological recursion to WKB asymptotic solutions of differential equations. In matrix models one can construct an object $\psi$ called the wave-function of the theory. On general grounds, it is then expected that there exists a ``quantization'' $\hat{P}(\hat x, {\hat y}; \hbar)$ of the spectral curve that annihilates $\psi$; this quantization of the spectral curve is generally called a \emph{quantum curve}. 

But what do we mean by quantization here? Assume that the spectral curve $P(x,y) = 0$ has degree $d$ in $y$. Let $\hat{P}(\hat x, {\hat y}; \hbar)$ be a polynomial in $\hat x$ and $\hat y$ of degree $d$ in $\hat y$, with coefficients that are possibly power series in $\hbar$. We let $\hat x$ and $\hat y$ be quantizations of the variables $x$ and $y$:
\begin{equation}
\hat x = x, \qquad \hat y = \hbar \frac{d}{dx},
\end{equation}
such that $[ \hat y, \hat x] =  \hbar$. This makes $\hat{P}(\hat x, {\hat y}; \hbar)$ into an order $d$ differential operator in $x$, with coefficients that are polynomial in $x$ and possibly power series in $\hbar$. We then say that $\hat{P}(\hat x, {\hat y}; \hbar)$ is a \emph{quantum curve} (of the original spectral curve) if, after normal ordering (that is bringing all the $\hat x$'s to the left of the $\hat y$'s), it takes the form
\begin{equation}
\hat{P}(\hat x, {\hat y}; \hbar) = P(\hat x, \hat y) + \sum_{n \geq 1} \hbar^n P_n(\hat x, \hat y),
\end{equation}
where the leading order term $P(\hat x, \hat y)$ recovers the original spectral curve (normal ordered), and the $P_n(\hat x, \hat y)$ are differential operators in $x$ of order at most $d-1$.

With this definition, the expectation from matrix models is that the wave-function $\psi$ should be the WKB asymptotic solution of the differential equation
\begin{equation}
\hat{P}(\hat{x}, \hat{y}; \hbar) \psi  = 0,
\end{equation}
for some quantum curve $\hat{P}(\hat x, {\hat y}; \hbar)$.
 This expectation follows from determinantal formulae in matrix models \cite{BeE,BeE2}.

The question of existence of quantum curves can be explored without reference to matrix models. Indeed, one can construct a wave-function $\psi$ in a natural way from the meromorphic differentials $W_{g,n}$ obtained from the topological recursion, without reference to any underlying matrix model. The question is then: for arbitrary spectral curves, does there exist a quantization that kills the wave function? And if so, can we construct this quantum curve explicitly from the topological recursion?

It should be noted here that this connection also has a deep relation with integrable systems. As explained in \cite{BorE1}, one can think of the wave-function as the Schlesinger transform of the partition function of the theory. If one assumes that the partition function is a $\tau$-function, \emph{i.e.} that it satisfies the Hirota equations, then it follows that the wave-function should be annihilated by a quantization of the spectral curve. However, it is not known whether the partition function constructed from the topological recursion is a $\tau$-function in general. 

An answer to the question above about the existence of quantum curves was provided in \cite{BE} for a very large class of genus zero spectral curves. More precisely, using the global topological recursion constructed in \cite{BE:2012, BHLMR}, it was proven that there exists a quantization that kills the wave-function for all spectral curves whose Newton polygons have no interior point and that are smooth as affine curves. For any such spectral curve, the quantum curve was reconstructed explicitly from the topological recursion. In fact, the quantum curve is not unique; one obtains different quantum curves (corresponding to different choices of ordering) depending on how one integrates the $W_{g,n}$ to construct the wave-function. More details can be found in \cite{BE}.

The goal of this paper is to study the relation between the topological recursion and quantum curves for genus one spectral curves. More precisely, we will focus on the family of spectral curves given by the Weierstrass equation
\begin{equation}
y^2 = 4 x^3 - g_2(\tau) x - g_3(\tau).
\end{equation}
We can apply the Eynard-Orantin topological recursion to this family of curves; the meromorphic differentials $W_{g,n}$ are elliptic and quasi-modular, while the free energies $F_g$ obtained from the recursion are quasi-modular forms. An interesting open question is whether these $W_{g,n}$ and $F_g$ have an enumerative interpretation for some geometric problem. We do not have an answer to this question. Nevertheless, this spectral curve is an interesting playground to study the connection between the topological recursion and quantum curves for spectral curves of genus greater than zero, since everything can be calculated very explicitly in terms of Weierstrass $\wp$ and $\wp'$ functions and Eisenstein series.

We initially study the wave-function $\psi$ constructed as in \cite{BE}; this is known as the \emph{perturbative wave-function}. For the Weierstrass spectral curve, it is obtained from the $W_{g,n}$ through the standard equation
\begin{multline}
\psi(z) =  \exp \left( \frac{1}{\hbar} \sum_{2g-1+n \geq 0} \frac{\hbar^{2g+n-1}}{n!} \int_0^z \cdots \int_0^z \Big(W_{g,n}(z_1,\ldots,z_n) \right.\\ \left. \left.  - \delta_{g,0}\delta_{n,2} \frac{dx(z_1) dx(z_2)}{(x(z_1) - x(z_2))^2} \right) \right).
\end{multline}
For this $\psi(z)$ we follow the steps in \cite{BE}, suitably generalized for our genus one curve, and construct an order two differential operator that annihilates $\psi$; however, it is not a quantum curve as defined above. But this was to be expected; from matrix model theory, when the spectral curve has genus greater than zero, the right wave-function to consider is not the perturbative wave-function. Rather, it needs to be corrected non-perturbatively.

A non-perturbative partition function was defined in \cite{E2,EM1} from the topological recursion directly, without reference to matrix models. The idea was to make the partition function modular invariant, which requires non-perturbative corrections. From the non-perturbative partition function one can define a non-perturbative wave-function as the Schlesinger transform \cite{BoE, BorE1}. This non-perturbative wave-function is the object that should be annihilated by a quantum curve.

We study the non-perturbative wave-function for the Weierstrass spectral curve. To obtain a well defined power series expansion in $\hbar$, it turns out that a quantization condition must be satisfied. The simplest elliptic curve that satisfies the quantization condition is:
$$
y^2 = 4 (x^3-1),
$$
which is of course a very special elliptic curve --- for instance, its $j$-invariant vanishes. 
Focusing on this spectral curve, through extensive symbolic calculations on Mathematica we calculate the wave-function up to order $5$ in $\hbar$. Remarkably, while the $W_{g,n}$ become extremely complicated, they somehow combine into very nice and simple elliptic functions in the non-perturbative wave-function. Using these calculations we are able to construct a quantum curve that annihilates the non-perturbative wave-function up to order $\hbar^5$. Perhaps unexpectedly, the quantization of the spectral curve includes non-trivial $\hbar$ corrections --- in fact probably an infinite number of such corrections. Nonetheless, the quantum curve is a true quantization of the spectral curve according to the definition above, as suggested by matrix model arguments. Therefore we obtain a proof of the existence of the quantum curve for the non-perturbative wave-function up to order $\hbar^5$, for this particular elliptic curve.

Going back to the perturbative wave-function, we perform the calculation of the differential operator in two different ways. It turns out that equivalence of the two approaches implies an infinite sequence of identities for $A$-cycle integrals of elliptic functions with quasi-modular properties. In particular, an infinite sub-sequence relates $A$-cycle integrals of elliptic functions to quasi-modular forms. In this paper, we write down explicitly only the first few identities, but it would certainly be interesting to study whether these identities are interesting from the point of view of elliptic functions and quasi-modular forms. For instance, they may be related to the results on quasi-modular forms obtained in \cite{GM}. We hope to report on that in the near future.

\subsection*{Outline}

In Section \ref{s:background} we review background material on elliptic functions and quasi-modular forms that will be needed in this paper. We also define the Eynard-Orantin topological recursion. In Section \ref{s:one}, we generalize the approach of \cite{BE} to construct an order two differential operator that annihilates the perturbative wave-function. In Section \ref{s:two}, we construct the same differential operator using a different approach, via the Riemann bilinear identity. We explore the connection and equivalence between the two approaches, which leads to the proof of an infinite sequence of identities for elliptic functions and quasi-modular forms. We briefly explore the first few of these identities in Section \ref{s:identities}. Then in Section \ref{s:NP} we study the non-perturbative wave-function and construct a quantum curve up to order $\hbar^5$. We conclude in Section \ref{s:conclusion}, with open questions and research avenues. Finally, we record in Appendix \ref{s:correlation} the first few $W_{g,n}$ constructed from the Weierstrass spectral curve, and in Appendix \ref{s:proof} we provide an independent proof of the simplest identity  that we obtained for elliptic functions, without reference to the topological recursion.

\subsection*{Acknowledgments}

We would like to thank G.~Borot, T.~Bridgeland, M. Mari\~no, M. Mulase, P.~Norbury and N.~Orantin for interesting discussions, and especially B.~Eynard for many enlightening discussions on related topics in recent years. We would also like to thank M.~ Westerholt-Raum for pointing out the paper \cite{GM} to us, and the referee for insightful comments. V.B. acknowledges
the support of the Natural Sciences and Engineering Research Council of Canada.

\section{Topological recursion on Weierstrass spectral curve}
\label{s:background}

\subsection{Elliptic functions and modular forms}

We start by defining standard objects that will be useful in this paper.

Let $\mathcal{H} = \{ \tau \in \mathbb{C} | \Im(\tau) > 0 \}$ be the upper half-plane, and define the lattice $\Lambda = \mathbb{Z} + \mathbb{Z} \tau$. The quotient $\mathbb{C} / \Lambda$ is topologically a torus. Functions on $\mathbb{C}/ \Lambda$ are given by doubly periodic functions, known as elliptic functions.

The Weierstrass function $\wp(z;\tau)$ is an example of an elliptic function. It is defined by
\begin{equation}
\wp(z;\tau) = \frac{1}{z^2} + \sum_{\omega \in \Lambda^*} \left(\frac{1}{(z-\omega)^2}	- \frac{1}{\omega^2}  \right),
\end{equation}
where $\Lambda^* = \Lambda \setminus \{ 0 \}$. The Weierstrass function $\wp'(z;\tau)$ is the derivative of $\wp(z;\tau)$ with respect to $z$; it is given by
\begin{equation}
\wp'(z;\tau) = -2 \sum_{\omega \in \Lambda} \frac{1}{(z-\omega)^3}.
\end{equation}
We also define the Weierstrass function $\zeta(z;\tau)$ by
\begin{equation}
\zeta(z; \tau) = \frac{1}{z} + \sum_{\omega \in \Lambda^*} \left(\frac{1}{z -\omega}  + \frac{1}{\omega} + \frac{z}{\omega^2} \right).
\end{equation}
This function is not elliptic; rather, it is quasi-elliptic, since $\zeta(z+1; \tau) = \zeta(z;\tau) + 2 \zeta(1/2; \tau)$ and $\zeta(z + \tau; \tau) = \zeta(z;\tau) + 2 \zeta(\tau/2; \tau)$.
It is clear that
\begin{equation}
\wp(z;\tau) = - \frac{d}{dz} \zeta(z;\tau).
\end{equation}

The Eisenstein series $G_{2n}(\tau)$, for $n \geq 2$, are defined by the uniformly convergent series\begin{equation}
G_{2n}(\tau) = \sum_{\omega \in \Lambda^*} \frac{1}{\omega^{2n}}.
\end{equation}
They are weight $2n$ modular forms, which means that they transform as
\begin{equation}
G_{2n}\left(\gamma \tau \right) = (c \tau + d)^{2n} G_{2n}(\tau), \quad \forall \gamma= \begin{pmatrix}a & b \\ c & d \end{pmatrix} \in SL(2,\mathbb{Z}),
\end{equation}
with
\begin{equation}
\gamma \tau = \frac{a \tau + b}{c \tau + d}.
\end{equation}

We can extend this definition to $n=1$, but then the series is not absolutely convergent anymore, so the order of summation matters. We define the second Eisenstein series $G_2(\tau)$
\begin{equation}
G_2(\tau) =\sum_{m \neq 0} \frac{1}{m^2} + \sum_{n \neq 0} \sum_{m \in \mathbb{Z}} \frac{1}{(m + n \tau)^2}.
\end{equation}
Because of the non-absolute convergence, we cannot change the order of summation, and it can be shown that it implies that $G_2(\tau)$ is not a modular form, but is rather a quasi-modular form of weight $2$. This means that it transforms with a shift, as
\begin{equation}\label{eq:G2trans}
G_2\left( \gamma \tau \right) = (c \tau + d)^2 G_2(\tau) -  2 \pi ic (c \tau + d ).
\end{equation}

We define the invariants
\begin{equation}
g_2(\tau) = 60 G_4(\tau), \qquad g_3(\tau) = 140 G_6(\tau). 
\end{equation}
It is well known that the Weierstrass functions satisfy the equation
\begin{equation}\label{eq:WPdiffeq}
\wp'(z;\tau)^2 = 4 \wp(z;\tau)^3 - g_2(\tau) \wp(z;\tau) - g_3(\tau).
\end{equation}
In other words, $x= \wp(z;\tau)$ and $y=\wp'(z;\tau)$ parameterize the Weierstrass curve
\begin{equation}
y^3 = 4 x^3 - g_2(\tau) x - g_3(\tau).
\end{equation}

The Weierstrass $\wp'(z;\tau)$ has three simple zeros at the half-periods
\begin{equation}
w_1 = \frac{1}{2}, \qquad w_2 = \frac{\tau}{2}, \qquad w_3 = - \frac{1}{2} (1 + \tau).
\end{equation}
As is customary, we denote by
\begin{equation}
e_1 = \wp(w_1; \tau), \qquad e_2 = \wp(w_2; \tau), \qquad e_3 = \wp(w_3; \tau)
\end{equation}
the value of the $\wp(z;\tau)$ function at the half-periods. We introduce the discriminant
\begin{equation}
\Delta(\tau) = g_2(\tau)^3 - 27 g_3(\tau)^2 = 16 (e_1-e_2)^2 (e_2-e_3)^2 (e_3-e_1)^2,
\end{equation}
which is a modular form of weight $12$.

The Weierstrass $\wp(z;\tau)$ function has a double pole at $z=0$. Its expansion near $z=0$ has a nice form; it is given by
\begin{equation}\label{eq:wpexpansion}
\wp(z;\tau) = \frac{1}{z^2} + \sum_{k=1}^\infty (2k+1) G_{2k+2}(\tau) z^{2k}.
\end{equation}
Following, for instance, \cite{Mason:2010}, let us define a new function $P_2(z;\tau)$ by including the $k=0$ term in the sum above:
\begin{equation}\label{eq:P2}
P_2(z;\tau) = \wp(z;\tau) + G_2(\tau).
\end{equation}
Of course $P_2(z;\tau)$ is not modular anymore, it is quasi-modular, because of $G_2(\tau)$. It is straightforward to show that it can be rewritten as
\begin{equation}
P_2(z; \tau) = (2 \pi i)^2 \sum_{n \in \mathbb{Z}^*} \frac{n q_z^n}{1 - q^n},
\end{equation}
where $q_z = \exp(2 \pi i z)$ and $q = \exp(2 \pi i \tau)$.

We also introduce $P_1(z; \tau)$ such that $\frac{d P_1(z;\tau)}{d z} = P_2(z;\tau)$. It follows that
\begin{equation}
P_1(z; \tau) = 2 \pi i \sum_{n \in \mathbb{Z}^*} \frac{q_z^n}{1-q^n} + A,
\end{equation}
for some constant $A$. We fix $A$ such that $P_1(-z; \tau) = - P_1(z; \tau)$. It follows that $A = \pi i$, that is,
\begin{equation}
P_1(z; \tau) = 2 \pi i \left(\sum_{n \in \mathbb{Z}^*} \frac{q_z^n}{1-q^n} + \frac{1}{2} \right).
\end{equation}
In terms of standard elliptic functions, we get
\begin{equation}
P_1(z;\tau) = - \zeta(z;\tau) + G_2(\tau) z .
\end{equation}
$P_1(z; \tau)$ is not elliptic anymore, but its transformation properties can be calculated. It is straightforward to show that
\begin{equation}
P_1(z + 1; \tau) = P_1(z), \qquad P_1(z + \tau; \tau) = P_1(z) + 2 \pi i.
\end{equation}
The second transformation is of course what makes it not quite elliptic.

\subsection{Spectral curve}

To define the topological recursion we need to introduce the notion of spectral curve.

\begin{defn}\label{d:sc}
A \emph{spectral curve} is a triple $(\Sigma,x,y)$ where $\Sigma$ is a Torelli marked genus $\hat{g}$ compact Riemann surface\footnote{A Torelli marked compact Riemann surface $\Sigma$ is a genus $\hat{g}$ Riemann surface $\Sigma$ with a choice of symplectic basis of cycles $(A_1, \ldots, A_{\hat{g}}, B_1, \ldots, B_{\hat{g}}) \in H_1(\Sigma,\mathbb{Z})$.} and $x$ and $y$ are meromorphic functions on $\Sigma$, such that the zeros of $dx$ do not coincide with the zeros of $dy$. 
\end{defn}

In this paper we will focus on one particular family of spectral curves, which we call the Weierstrass spectral curve.

\begin{defn}\label{d:wsc}
The \emph{Weierstrass spectral curve} is defined  by the triple $(\Sigma, x, y)$, where, $\Sigma = \mathbb{C}/\Lambda$ with lattice $\Lambda = \mathbb{Z} \oplus \tau \mathbb{Z}$, $x=\wp(z; \tau)$ and $y=\wp'(z;\tau)$. Then the meromorphic functions $x$ and $y$ identically satisfy the Weierstrass equation
\begin{equation}
y^2 = 4 x^3 - g_2(\tau) x -  g_3( \tau).
\end{equation}
\end{defn}

This is of course a genus one spectral curve, since $\Sigma$ is a torus. In fact, it is a family of curves, parametrized by the complex modulus $\tau$.

As usual in topological recursion we are interested in the branched covering $\pi : \Sigma \to \mathbb{P}^1$ given by the meromorphic function $x$. For the Weierstrass spectral curve, $\pi$ is a double cover. The deck transformation that exchanges the two sheets is simply given by $z \mapsto \tau(z) = - z$, since $\wp(z;\tau)$ is an even function in $z$.

We denote by $R$ the set of ramification points of $\pi$, which is given by the zeros of $dx$ and the poles of $x$ of order $\geq 2$. For the Weierstrass spectral curve, since
\begin{equation}
dx = \wp'(z;\tau) d z,
\end{equation}
the zeros of $dx$ are given by the half-periods $w_i$, $i=1,2,3$ introduced earlier. Moreover, $x = \wp(z;\tau)$ has a double pole at $z=0$. Therefore $R = \{w_1, w_2, w_3, 0\}$.

\subsection{Geometric objects}

For the topological recursion we also need the following objects that are canonically defined on a genus $\hat{g}$ compact Riemann surface $\Sigma$ with a symplectic basis of cycles for $H_1(\Sigma,\mathbb{Z}$).

\begin{defn}\label{d:w}
Let $a,b \in \Sigma$. The \emph{canonical differential of the third kind} $\omega^{a-b}(z)$ is a meromorphic one-form on $\Sigma$ such that:
\begin{itemize}
\item it is holomorphic away from $z=a$ and $z=b$;
\item it has a simple pole at $z=a$ with residue $+1$;
\item it has a simple pole at $z=b$ with residue $-1$;
\item it is normalized on $A$-cycles:
\begin{equation}
\oint_{z \in A_i} \omega^{a-b}(z) =  0, \qquad \text{for $i=1,\ldots,\hat{g}$.}
\end{equation}
\end{itemize}
\end{defn}

\begin{defn}\label{d:B}
The \emph{canonical bilinear differential of the second kind} $B(z_1,z_2)$ is the unique bilinear differential on $\Sigma^2$ satisfying the conditions:
\begin{itemize}
\item It is symmetric, $B(z_1, z_2) = B(z_2,z_1)$;
\item It has its only pole, which is double, along the diagonal $z_1 = z_2$, with leading order term (in any local coordinate $z$)
\begin{equation}
B(z_1,z_2) \underset{z_1 \to z_2 }{\rightarrow} \frac{d z_1 d z_2}{(z_1-z_2)^2} + \ldots;
\end{equation}
\item It is normalized on $A$-cycles:
\begin{equation}
\oint_{z_1 \in A_i} B(z_1, z_2) = 0, \qquad \text{for $i=1,\ldots,\hat{g}$}.
\end{equation}
\end{itemize}
\end{defn}

\begin{rem}
It follows from the definition that
\begin{equation}
B(z_1, z_2) = d_1 \omega^{z_1-b}(z_2).
\end{equation}
Equivalently,
\begin{equation}\label{e:omega}
\omega^{a - b}(z) = \int_{z_1=b}^a B(z_1, z),
\end{equation}
where the integral is taken over the unique homology chain with boundary $[a]-[b]$ that doesn't intersect the homology basis.
\end{rem}

It is not too difficult to identify what these objects are on the Weierstrass spectral curve. Recall that $\Sigma = \mathbb{C} / \Lambda$ with lattice $\Lambda = \mathbb{Z} \oplus \tau \mathbb{Z}$.  We fix the A-cycle to be given by $z \in [0,1)$,\footnote{To be precise, to integrate elliptic functions over $A$-cycles we need to shift it by $i \epsilon$ to avoid poles on the contour. For the same reason, a similar shift by a purely real $\epsilon$ must be done when evaluating $B$-cycle integrals.} and the B-cycle to be given by $z \in [0,\tau)$. The canonical bilinear differential of the second kind is given by
\begin{equation}
B(z_1,z_2) = P_2(z_1 - z_2; \tau) d z_1 d z_2,
\end{equation}
where $P_2(z;\tau)$ was introduced in \eqref{eq:P2}. First, it is symmetric, since $P_2(z;\tau)$ is an even function of $z$. Second, it is clear from \eqref{eq:wpexpansion} that it has a double pole on the diagonal with the right leading behavior. It has no other poles. As for normalization, one can check that it is normalized on the A-cycle.

The canonical differential of the third kind is then given by
\begin{equation}
\omega^{a-b}(z) =\left( P_1(z-b;\tau) - P_1(z-a;\tau) \right) dz.
\end{equation}

\subsection{Topological recursion}

We now introduce the topological recursion formalism, which was first proposed in \cite{CEO,E1,EO,EO2}. For clarity of presentation, we will only introduce the formalism in the context of the Weierstrass spectral curve. 

Let $(\Sigma,x,y)$ be a spectral curve. The topological recursion constructs an infinite tower of symmetric meromorphic differentials $W_{g,n}(z_1,\ldots,z_n)$, known as \emph{correlation functions}, on $\Sigma^n$. We now consider the special case where $(\Sigma,x,y)$ is the Weierstrass spectral curve.

\begin{defn}[Topological recursion]\label{d:TR}
We first define the initial conditions
\begin{align}
W_{0,1}(z) =& y(z) dx(z) = \wp'(z;\tau)^2 dz, \\ W_{0,2}(z_1,z_2) =& B(z_1,z_2) = P_2(z_1-z_2; \tau) dz_1 dz_2.
\end{align}

Let $\mathbf{z} = \{z_1,\ldots,z_n\}\in \Sigma^n$.
For $n \geq 0$, $g \geq 0$ and $2g-2+n \geq 0$, we uniquely construct symmetric meromorphic differentials $W_{g,n}$ on $\Sigma^n$ with poles along $R$ via the \emph{Eynard-Orantin topological recursion}:
\begin{equation}\label{eq:TR}
W_{g,n+1}(z_0, \mathbf{z}) =  \sum_{a \in R} \underset{z = a}{\text{Res}}~ K(z_0; z) \mathcal{R}^{(2)} W_{g,n+1}(z,-z; \mathbf{z}),
\end{equation}
where the recursion kernel is given by 
\begin{equation}
K(z_0; z) = \frac{\omega^{z - \alpha}(z_0)}{( y(z) - y(-z) )dx(z)} = \frac{\left(P_1(z_0-\alpha;\tau) - P_1(z_0-z;\tau) \right) d z_0}{2 \wp'(z;\tau)^2 d z},
\end{equation}
with $\alpha$ an arbitrary base point on $\Sigma$ (it can be checked that the definition is independent of the choice of $\alpha$, as a long as it is generic). The recursive structure is given by
\begin{multline}
 \mathcal{R}^{(2)} W_{g,n+1}(z,-z; \mathbf{z}) \\=W_{g-1,n+2}(z,-z,\mathbf{z}) + \sum_{g_1+g_2 =g} \sum_{I \cup J = \mathbf{z}}' W_{g_1, |I|+1}(z,I) W_{g_2, |J|+1}(-z,J).
 \label{eq:R2Wgn}
\end{multline}
In the second sum we are summing over all disjoint subsets $I, J \subset \mathbf{z}$ whose union is $\mathbf{z}$, and the prime means that we exclude the cases $(g_1, |I|) = (0,0)$ and $(g_2,|J|) = (0,0)$.
\end{defn}

We can compute the first few correlation functions explicitly for the Weierstrass spectral curve. Those are presented in Appendix A.

For later reference we also introduce the free energies $F_g := W_{g,0}$, $g \geq 2$, for the spectral curve. Those are obtained via the auxiliary equation
\begin{equation}\label{eq:fg}
F_g = \frac{1}{2-2g} \sum_{a \in R}  \underset{z = a}{\text{Res}}~ \phi(z) W_{g,1}(z),
\end{equation}
with $\phi(z) = \int y(z) dx(z)$ an arbitrary antiderivative of the one-form $y dx$.

\subsection{Quantum curve}

The purpose of this paper is to related the meromorphic differentials $W_{g,n}$ constructed from the topological recursion to the WKB asymptotic solution of a differential operator, known as a quantum curve. The connection will be explored in more detail later on in this paper, but for clarity and further reference let us define here what we mean by quantum curve.

Consider a spectral curve $(\Sigma, x,y)$. Let $P(x,y) = 0$ be the minimal irreducible polynomial equation satisfied by $x$ and $y$. Assume that it has degree $d$ in $y$. 

Define the quantization of the variables $x$ and $y$ as:
\begin{equation}
\hat x = x, \qquad \hat y = \hbar \frac{d}{dx},
\end{equation}
such that $[ \hat y, \hat x] =  \hbar$.

\begin{defn}[Quantum curve]\label{d:qc}
Let $\hat{P}(\hat x, {\hat y}; \hbar)$ be a polynomial in $\hat x$ and $\hat y$ of degree $d$ in $\hat y$, with coefficients that are possibly power series in $\hbar$. $\hat{P}(\hat x, {\hat y}; \hbar)$ is an order $d$ differential operator in $x$, with polynomial coefficients in $x$ that are possibly power series in $\hbar$. We say that $\hat{P}(\hat x, {\hat y}; \hbar)$ is a \emph{quantum curve} (of the original spectral curve) if, after normal ordering (that is bringing all the $\hat x$'s to the left of the $\hat y$'s), it takes the form
\begin{equation}
\hat{P}(\hat x, {\hat y}; \hbar) = P(\hat x, \hat y) + \sum_{n \geq 1} \hbar^n P_n(\hat x, \hat y),
\end{equation}
where the leading order term $P(\hat x, \hat y)$ recovers the original spectral curve (normal ordered), and the $P_n(\hat x, \hat y)$ are differential operators in $x$ of order at most $d-1$.
\end{defn}

We note here that the requirement that the $P_n(\hat x, \hat y)$ have order at most $d-1$ is equivalent to requiring that the coefficient of the highest degree $\hat y^d$ term in $\hat{P}(\hat x, {\hat y}; \hbar)$ does not depend on $\hbar$.

As an example that will be particularly relevant later on, consider the elliptic spectral curve $y^2 = 4  (x^3-1)$. 
According to the definition above, a quantum curve for this spectral curve must take the form
\begin{equation}
\hat{P}(\hat{x},\hat{y};\hbar) = \hbar^2 \frac{\text{d}^2}{\text{d} x^2}  - 4 (x^3-1) + {\sum_{i \geq 2} \hbar^{i} A_{i} (x) \frac{\text{d}}{\text{d} x} } + {\sum_{j \geq 1} \hbar^{j} B_{j} (x)},
\label{eq:qcwe}
\end{equation}
with the $A_i(x)$ and $B_j(x)$ polynomials in $x$.

\section{Perturbative wave-function: first approach}

\label{s:one}

In this section we approach the problem of constructing the quantum curve for the Weierstrass spectral curve naively. We apply the method of \cite{BE} directly with a few modifications. The idea here is to construct the wave-function as in \cite{BE}, which is what we will call the ``perturbative wave-function'', and then show that it is annihilated by an order two differential operator. However, this differential operator is not a quantum curve, according to Definition \ref{d:qc}. But this is because, as we will see in section \ref{s:NP}, and as is already expected from matrix models (see for instance \cite{BoE}), the perturbative wave-function is not the right object to look at. For spectral curves of genus $\geq 1$, one should use the non-perturbative wave-function to construct the quantum curve. We will study that in more detail in section \ref{s:NP}.

Coming back to the goal of this section, recall that in \cite{BE} quantum curves were obtained for all spectral curves whose Newton polygons have no interior point and that are smooth as affine curves. Of course, the Weierstrass spectral curve does not fall into that class, since its Newton polygon has an interior point (it has genus one). However, the main results of \cite{BE} can be adapted for this particular case, which is what we do in this section.

In this section we borrow heavily on the notation and calculations of \cite{BE}, even though the calculations are much simpler in the case at hand. The reader may want to refer to this paper for more detail.

\subsection{Reconstructing loop equations}

The first step in \cite{BE} is to reconstruct some sort of ``loop equations'' from the topological recursion. This is the content of Lemma 4.7 in \cite{BE}. Let us recall the notation. Here we focus on the Weierstrass spectral curve with the branched covering $\pi$. Since $\pi$ is a double cover and the deck transformation is given by $z \mapsto - z$, the objects introduced in \cite{BE} simplify drastically.

We first define
\begin{multline}
Q^{(2)}_{g,n+1}(z;\mathbf{z})\\= W_{g-1,n+2}(z,-z,\mathbf{z}) + \sum_{g_1+g_2 =g} \sum_{I \cup J = \mathbf{z}} W_{g_1, |I|+1}(z,I) W_{g_2, |J|+1}(-z,J).
\end{multline}
This is just like the recursive structure $ \mathcal{R}^{(2)} W_{g,n+1}(z,-z; \mathbf{z}) $ introduced in \eqref{eq:R2Wgn}, but with the $(g_1,|I|) = (0,0)$ and $(g_2,|J|) = (0,0)$ terms included. 

In our context, Lemma 4.7 of \cite{BE} becomes the following statement:
\begin{lem}
For $2g-2+n \geq 0$, the meromorphic one-forms (in $z$)
\begin{equation}
dz \left( \frac{Q^{(2)}_{g,n+1}(z;\mathbf{z})}{dx(z)^2} \right).
\end{equation}
can only have poles (in $z$) at $z = \pm z_i$, $i=1,\ldots,n$.
\end{lem}

\begin{proof}
The proof of Lemma 4.7 presented in \cite{BE} only requires the spectral curve to be smooth as an affine curve; it does not require the property that the Newton polygon has no interior point. Since the Weierstrass spectral curve is generically smooth as an affine curve, the proof goes through untouched.
\end{proof}

The next step in the approach of \cite{BE} is to prove Lemma 4.8, Lemma 4.9 and Theorem 4.12. Here the proofs need to be modified, and the results will differ. So let us go through these statements carefully.

The first lemma is
\begin{lem}
For the Weierstrass spectral curve,
\begin{equation}
\frac{Q^{(2)}_{0,1}(z)}{dx(z)^2} = - 4 x(z)^3 + g_2(\tau)x(z) + g_3(\tau).
\end{equation}
\end{lem}

\begin{proof}
This is straightforward since
\begin{equation}
\frac{Q^{(2)}_{0,1}(z)}{dx(z)^2} = \frac{W_{0,1}(z) W_{0,1}(-z)}{dx(z)^2} = y(z) y(-z) = - y(z)^2.
\end{equation}
\end{proof}

The second lemma is a little more involved:
\begin{lem}\label{l:sec}
For the Weierstrass spectral curve,
\begin{equation}
\frac{Q^{(2)}_{0,2}(z; z_1)}{dx(z)^2} = - d_{z_1} \left(\frac{1}{(x(z) - x(z_1)) } \frac{W_{0,1}(z_1)}{dx(z_1)} \right)- 2 P_2(z_1; \tau) d z_1.
\end{equation}
\end{lem}

\begin{proof}
Here the proof from \cite{BE} needs to be modified, so let us do it carefully. We have:
\begin{align}
\frac{Q^{(2)}_{0,2}(z; z_1)}{dx(z)^2} =& \frac{B(z,z_1)}{dx(z)} \frac{W_{0,1}(-z)}{dx(z)} +  \frac{B(-z,z_1)}{dx(z)} \frac{W_{0,1}(z)}{dx(z)} \nonumber\\
=& - \frac{B(z,z_1)}{dx(z)} \frac{W_{0,1}(z)}{dx(z)} -  \frac{B(-z,z_1)}{dx(z)} \frac{W_{0,1}(-z)}{dx(z)} \nonumber\\
=&- \underset{z'=z}{\text{Res}}~ \frac{B(z',z_1)}{x(z')-x(z)} \frac{W_{0,1}(z')}{dx(z')} - \underset{z'=-z}{\text{Res}}~ \frac{B(z',z_1)}{x(z')-x(z)} \frac{W_{0,1}(z')}{dx(z')} .
\end{align}
Now the expression
\begin{equation}
\frac{B(z',z_1)}{x(z')-x(z)} \frac{W_{0,1}(z')}{dx(z')}
\end{equation}
is a meromorphic one-form in $z'$ on the compact Riemann surface $\Sigma$. Therefore, the sum of its residues must vanish. Its only poles are at $z'=z_1$, $z' = \pm z$, and at the pole of $\frac{W_{0,1}(z')}{dx(z')} = y(z')$, that is, at $z'=0$.\footnote{This is where the proof differs from \cite{BE}. When the Newton polygon has no interior point, the only poles are at $z'=z_1$ and at $z' = \tau_i(z)$ where $\tau_i(z)$ indexes the sheets of the branched covering $\pi$.} Thus we get
\begin{align}
\frac{Q^{(2)}_{0,2}(z; z_1)}{dx(z)^2} =& -\underset{z'=z_1}{\text{Res}}~ \frac{B(z',z_1)}{x(z)-x(z')} \frac{W_{0,1}(z')}{dx(z')}-  \underset{z'=0}{\text{Res}}~ \frac{B(z',z_1)}{x(z)-x(z')} \frac{W_{0,1}(z')}{dx(z')}\nonumber\\
=& -d_{z_1} \left(\frac{1}{x(z)-x(z_1)} \frac{W_{0,1}(z_1)}{dx(z_1)} \right) - 2 P_2(z_1; \tau) dz_1,
\end{align}
where we used the fact that $\wp(z;\tau) \sim \frac{1}{z^2}$ near $z=0$, $\wp'(z;\tau) \sim - \frac{2}{z^3}$, and $B(z_1,z_2) = P_2(z_1-z_2;\tau) dz_1 dz_2$.
\end{proof}

\begin{rem}
Note that by replacing both the left-hand-side and right-hand-side of Lemma \ref{l:sec} in terms of elliptic functions, one can show that the statement above reduces to the well known identity:
\begin{multline}
P_2(z-z_1;\tau) + P_2(z+z_1;\tau)  \\= \frac{\wp''(z_1;\tau)}{\wp(z;\tau)-\wp(z_1;\tau)} + \frac{\wp'(z_1;\tau)^2}{(\wp(z;\tau)-\wp(z_1;\tau))^2} +2 P_2(z_1;\tau).
\end{multline}
\end{rem}

And finally, the main result that replaces Theorem 4.12 of \cite{BE} is the following theorem:
\begin{prop}\label{p:main1}
For the Weierstrass spectral curve, for $2g-2+n \geq 0$,
\begin{equation}
\frac{Q_{g,n+1}^{(2)}(z;\mathbf{z})}{dx(z)^2} =-  \sum_{i=1}^n d_{z_i} \left(\frac{1}{x(z)-x(z_i)} \frac{W_{g,n}(\mathbf{z} )}{dx(z_i)} \right) - 2 \left(\frac{W_{g,n+1}(z', \mathbf{z})}{dz'} \right)_{z'=0}.
\end{equation}
For $(g,n) = (0,1)$,
\begin{equation}
\frac{Q^{(2)}_{0,2}(z; z_1)}{dx(z)^2} = - d_{z_1} \left(\frac{1}{(x(z) - x(z_1)) }\frac{W_{0,1}(z_1)}{dx(z_1)}  \right)- 2 P_2(z_1; \tau) d z_1,
\end{equation}
while for $(g,n) = (0,0)$,
\begin{equation}
\frac{Q^{(2)}_{0,1}(z)}{dx(z)^2} = - 4 x(z)^3 + g_2(\tau)x(z) + g_3(\tau).
\end{equation}
\end{prop}

\begin{proof}
The cases $(g,n)=(0,0)$ and $(g,n)=(0,1)$ were proven in the two previous lemmas. So let us focus on $2g-2+n \geq 0$.

First, notice that we can write
\begin{align}
\frac{Q_{g,n+1}^{(2)}(z;\mathbf{z})}{dx(z)^2} = &\frac{y(z) dz}{dx(z)}\frac{Q_{g,n+1}^{(2)}(z;\mathbf{z})}{dx(z)^2}   \nonumber\\
= &\frac{1}{2} \left(- \frac{y(-z) dz}{dx(z)}\frac{Q_{g,n+1}^{(2)}(-z;\mathbf{z})}{dx(z)^2}+\frac{y(z) dz}{dx(z)}\frac{Q_{g,n+1}^{(2)}(z;\mathbf{z})}{dx(z)^2} \right) \nonumber\\
=&\frac{1}{2} \left(  \underset{z'=z}{\text{Res}}~ \frac{y(z') dz'}{x(z')-x(z)}\frac{Q_{g,n+1}^{(2)}(z';\mathbf{z})}{dx(z')^2} \right. \nonumber\\  & \left. \qquad +\underset{z'=-z}{\text{Res}}~ \frac{y(z') dz'}{x(z')-x(z)}\frac{Q_{g,n+1}^{(2)}(z';\mathbf{z})}{dx(z')^2}  \right).
\end{align}
The expression
\begin{equation}
\frac{y(z') dz'}{x(z')-x(z)}\frac{Q_{g,n+1}^{(2)}(z';\mathbf{z})}{dx(z')^2}  
\end{equation}
is a meromorphic one-form in $z'$ on a compact Riemann surface, hence the sum of its residues must be zero. But recall that the one-forms
\begin{equation}
dz' \left( \frac{Q_{g,n+1}^{(2)}(z';\mathbf{z})}{dx(z')^2} \right)
\end{equation}
can only have poles (in $z'$) at $z' = \pm z_i$, $i=1,\ldots,n$. Therefore, the expression
\begin{equation}
\frac{y(z') dz'}{x(z')-x(z)}\frac{Q_{g,n+1}^{(2)}(z';\mathbf{z})}{dx(z')^2}  
\end{equation}
can only have poles (in $z'$) at $z' = \pm z$, $z' = \pm z_i$, $i=1,\ldots,n$, and at the pole of $y(z')$, that is, at $z'=0$.\footnote{Again, this last pole at $z'=0$ does not occur when the spectral curve is such that its Newton polygon has no interior point. This is what makes the Weierstrass spectral curve different from the curves studied in \cite{BE}.} Thus we get
\begin{align}
\frac{Q_{g,n+1}^{(2)}(z;\mathbf{z})}{dx(z)^2} =& \frac{1}{2}\left( \sum_{i=1}^n  \underset{z'=\pm z_i}{\text{Res}}~ \frac{y(z') dz'}{x(z)-x(z')}\frac{Q_{g,n+1}^{(2)}(z';\mathbf{z})}{dx(z')^2}  \right.  \nonumber \\ & \qquad \left. + \underset{z'=0}{\text{Res}}~ \frac{y(z') dz'}{x(z)-x(z')}\frac{Q_{g,n+1}^{(2)}(z';\mathbf{z})}{dx(z')^2}  \right) \nonumber\\
=&\frac{1}{2} \sum_{i=1}^n \left(\underset{z'= z_i}{\text{Res}}~ \frac{y(z') dz'}{x(z)-x(z')}\frac{B(z',z_i) W_{g,n}(-z', \mathbf{z} \setminus \{z _i \})}{dx(z')^2}  \right. \nonumber\\ 
& \qquad \left.+\underset{z'= -z_i}{\text{Res}}~ \frac{y(z') dz'}{x(z)-x(z')}\frac{B(-z',z_i) W_{g,n}(z', \mathbf{z} \setminus \{z _i \})}{dx(z')^2} \right) \nonumber\\
& +\left(\frac{Q_{g,n+1}^{(2)}(z';\mathbf{z})}{dx(z')^2} \right)_{z'=0}\nonumber\\
=& - \sum_{i=1}^n d_{z_i} \left(  \frac{y(z_i) dz_i}{x(z)-x(z_i)}\frac{W_{g,n}(\mathbf{z})}{dx(z_i)^2} \right) - 2 \left(\frac{W_{g,n+1}(z';\mathbf{z})}{dz'}  \right)_{z'=0} \nonumber\\
=&- \sum_{i=1}^n d_{z_i} \left(  \frac{1}{x(z)-x(z_i)}\frac{W_{g,n}(\mathbf{z})}{dx(z_i)} \right) - 2 \left(\frac{W_{g,n+1}(z';\mathbf{z})}{dz'}  \right)_{z'=0} .
\end{align}
Here we used the fact that as $z' \to 0$, we have that
\begin{equation}
\frac{Q_{g,n+1}^{(2)}(z';\mathbf{z})}{dx(z')^2}  \to \frac{W_{0,1}(z') W_{g,n+1}(-z';\mathbf{z}) + W_{0,1}(-z') W_{g,n+1}(z';\mathbf{z})}{dx(z')^2},
\end{equation}
since all other terms vanish because $z'=0$ is a pole of $dx(z')$.

\end{proof}

The following corrollary then follows directly from the definition of $Q_{g,n+1}^{(2)}(z;\mathbf{z})$ and the fact that
\begin{equation}
W_{0,2}(z,z_1) + W_{0,2}(-z,z_1) = \frac{dx(z) dx(z_1)}{(x(z) - x(z_1))^2}.
\end{equation}
\begin{cor}\label{c:mess}
For $2g-2+n \geq 0$,
\begin{multline}
- \frac{W_{g-1,n+2}(-z,z,\mathbf{z})}{dx(z)^2} + \sum_{g_1+g_2 =g} \sum_{I \cup J = \mathbf{z}} \frac{W_{g_1, |I|+1}(-z,I)}{dx(z)} \frac{W_{g_2, |J|+1}(-z,J)}{dx(z)} \\-  \sum_{i=1}^n \left(\frac{ dx(z_i)}{(x(z)-x(z_i))^2}\frac{W_{g,n}(-z, \mathbf{z} \setminus \{ z_i \})}{ dx(z) } \right. \\ \left. -d_{z_i} \left(\frac{1}{x(z)-x(z_i)} \frac{W_{g,n}(-z_i, \mathbf{z} \setminus \{ z_i \} )}{dx(z_i)} \right)  \right)
+2 \left(\frac{W_{g,n+1}(-z', \mathbf{z})}{dz'} \right)_{z'=0} = 0.
\label{eq:mess}
\end{multline}
For $(g,n)=(0,1)$,
\begin{multline}
 2 \frac{W_{0,2}(-z,z_1)}{dx(z)} \frac{W_{0,1}(-z)}{dx(z)} - \frac{dx(z_1)}{(x(z)-x(z_1))^2} \frac{W_{0,1}(-z)}{dx(z)} \\+d_{z_1} \left(\frac{1}{(x(z) - x(z_1)) }\frac{W_{0,1}(-z_1)}{dx(z_1)}  \right)
 -2 P_2(z_1; \tau) d z_1 = 0,
\end{multline}
while for $(g,n) = (0,0)$,
\begin{equation}
2 \frac{W_{0,1}(-z) W_{0,1}(-z)}{dx(z)^2} - 4 x(z)^3 + g_2(\tau) x(z) + g_3(\tau) = 0.
\end{equation}
\end{cor}

\subsection{Integration}

The next step is to integrate the equation above. Following the notation in \cite{BE}, we choose the integration divisor to be $D = [z]- [0]$, since $z=0$ is the only pole of $x(z)$. While $z=0$ is in $R$, it is easy to show that the correlation functions do not have poles at $z=0$, therefore the integrals will converge. 

\begin{defn}
We define
\begin{equation}
G_{g,n+1}(z; \mathbf{z}) =  \int_0^{z_1} \cdots \int_0^{z_n} W_{g,n+1}(-z, z_1, \ldots, z_n).
\label{eq:Ggndef}
\end{equation}
Note that we are integrating in each variable $z_1, \ldots, z_n$, with base point $0$, but we are not integrating in the variable $z$.
\end{defn}

Now we can integrate Corollary \ref{c:mess} in $z_1, \ldots, z_n$. We get:

\begin{lem}
For $2g-2+n \geq 0$,
\begin{multline}
 -\left(\frac{\partial}{\partial x(z_{n+1}) }\frac{G_{g-1,n+2}(z;\mathbf{z}, z_{n+1})}{dx(z)} \right)_{z_{n+1}=z}\\+ \sum_{g_1+g_2 =g} \sum_{I \cup J = \mathbf{z}} \frac{G_{g_1, |I|+1}(z;I)}{dx(z)} \frac{G_{g_2, |J|+1}(z;J)}{dx(z)} \\-\sum_{i=1}^n \left(\frac{ 1}{x(z_i)-x(z)} \left( \frac{G_{g,n}(z_i; \mathbf{z} \setminus \{ z_i \} )}{dx(z_i)}-\frac{G_{g,n}(z; \mathbf{z} \setminus \{ z_i \})}{ dx(z) }  \right) \right)  
  \\+ 2 \left(\frac{G_{g,n+1}(z';\mathbf{z})}{dz'} \right)_{z'=0}= 0.
\end{multline}
For $(g,n) = (0,1)$,
\begin{multline}
 2 \frac{G_{0,2}(z;z_1)}{dx(z)} \frac{G_{0,1}(z)}{dx(z)} - \frac{1}{x(z_1) - x(z) } \left(\frac{G_{0,1}(z_1)}{dx(z_1)}- \frac{G_{0,1}(z)}{dx(z)} \right) \\
 \\-2 P_1(z_1; \tau) = 0,
\end{multline}
while for $(g,n) = (0,0)$,
\begin{equation}
2 \frac{G_{0,1}(z) G_{0,1}(z)}{dx(z)^2} - 4 x(z)^3 + g_2(\tau) x(z) + g_3(\tau) = 0.
\end{equation}
\end{lem}

\begin{proof}
The integration is straightforward; all we need to do is be careful with the base point $0$. 

For $2g-2+n \geq 0$, integrating the terms inside the summation $\sum_{i=1}^n$ gives rise to a term of the form
\begin{equation}
\sum_{i=1}^n \lim_{z_i = 0}\left(\frac{ 1}{x(z)-x(z_i)}\frac{G_{g,n}(z; \mathbf{z} \setminus \{ z_i \})}{ dx(z) }-\frac{1}{x(z)-x(z_i)} \frac{G_{g,n}(z_i; \mathbf{z} \setminus \{ z_i \} )}{dx(z_i)}   \right) .
\end{equation}
The first term clearly vanishes since $z_i=0$ is a pole of $x(z_i)$. As for the second term, it also vanishes, because $G_{g,n}(z_i; \mathbf{z} \setminus \{ z_i \} )$ cannot have a pole at $z_i=0$. Hence we get the expression in the Lemma.

For $(g,n)=(0,1)$, integrating in $z_1$ gives rise to a term of the form
\begin{multline}
\lim_{z_1 = 0} \left[ \frac{1}{x(z_1) - x(z) } \left(\frac{G_{0,1}(z_1)}{dx(z_1)}- \frac{G_{0,1}(z)}{dx(z)} \right) + 2 P_1(z_1;\tau) \right] \\
= \lim_{z_1 = 0} \left[ \frac{1}{x(z_1) - x(z) } \left(-y(z_1)+ y(z) \right) + 2 P_1(z_1;\tau) \right] .
\end{multline}
But since, as $z_1 \to 0$,
\begin{equation}
x(z_1) = \wp(z_1;\tau) \to \frac{1}{z_1^2}, \qquad y(z_1) = \wp'(z_1;\tau) \to - \frac{2}{z_1^3}, \qquad P_1(z_1;\tau) \to - \frac{1}{z_1},
\end{equation}
we see that the limit actually vanishes. Thus we get expression in the Lemma.

As for $(g,n)=(0,0)$, we are not integrating so the expression is obvious.
\end{proof}

\subsection{Principal specialization}

Then we ``principal specialize'' by setting $z_1 = \ldots = z_n = z$. We define
\begin{equation}
\hat{G}_{g,n+1}(z'; z) =G_{g,n+1}(z'; z, \ldots, z). 
\end{equation}
We get:
\begin{lem}\label{l:ps}
For $2g-2+n \geq 0$,
\begin{multline}
 -\frac{1}{n+1} \left(\frac{d}{d x(z) }\frac{\hat{G}_{g-1,n+2}(z';z)}{dx(z')} \right)_{z'=z}\\+ \sum_{g_1+g_2 =g} \sum_{m=0}^n \frac{n!}{m! (n-m)!} \frac{\hat{G}_{g_1, m+1}(z;z)}{dx(z)} \frac{\hat{G}_{g_2, n-m+1}(z;z)}{dx(z)} \\-n \left( \frac{d}{dx(z')}  \frac{\hat{G}_{g,n}(z'; z )}{dx(z')} \right)_{z'=z}     
 + 2 \left(\frac{\hat{G}_{g,n+1}(z';z)}{dz'} \right)_{z'=0}= 0.
\end{multline}
For $(g,n)=(0,1)$,
\begin{equation}
 2 \frac{\hat{G}_{0,2}(z;z)}{dx(z)} \frac{\hat{G}_{0,1}(z)}{dx(z)} - \frac{d}{dx(z)} \left( \frac{\hat{G}_{0,1}(z)}{dx(z)} \right) 
 -2 P_1(z; \tau) = 0,
\end{equation}
while for $(g,n) = (0,0)$,
\begin{equation}
2 \frac{\hat{G}_{0,1}(z) \hat{G}_{0,1}(z)}{dx(z)^2} - 4 x(z)^3 + g_2(\tau) x(z) + g_3(\tau) = 0.
\end{equation}

\end{lem}

\begin{proof}
This is straightforward, the only terms that need to be treated carefully are those that give rise to the derivatives.
\end{proof}

Finally, we sum over $g$ and $n$. More precisely, we define
\begin{equation}
\xi_1(z';z) = - \sum_{g,n=0}^\infty \frac{\hbar^{2g+n}}{n!} \frac{\hat{G}_{g,n+1}(z'; z) }{dx(z')}.
\end{equation}
Multiplying the equations in Lemma \ref{l:ps} by $\frac{\hbar^{2g+n}}{n!} $, and summing over $g$ and $n$, we get:
\begin{lem}
\begin{multline}\label{eq:xi}
\hbar \frac{d}{dx(z)} \xi_1(z;z) + \xi_1(z;z)^2 - 4 x(z)^3 + g_2(\tau) x(z) + g_3(\tau) \\
-2 \hbar P_1(z;\tau) + 2 \sum_{2g-2+n \geq 0} \frac{\hbar^{2g+n}}{n!} \left(\frac{\hat{G}_{g,n+1}(z';z)}{dz'} \right)_{z'=0} = 0.
\end{multline}
\end{lem} \qed

\subsection{Differential operator}

As in \cite{BE} we introduce the perturbative wave-function:
\begin{multline}\label{e:psi}
\psi(z) =     \exp \left( \frac{1}{\hbar} \sum_{2g-1+n \geq 0} \frac{\hbar^{2g+n-1}}{n!} \int_0^z \cdots \int_0^z \left(W_{g,n}(z_1,\ldots,z_n)   \right. \right. \\ \left. \left. - \delta_{g,0}\delta_{n,2} \frac{dx(z_1) dx(z_2)}{(x(z_1) - x(z_2))^2} \right) \right),
\end{multline}
and we define
\begin{equation}
\psi_1(z';z) = \psi(z) \xi_1(z';z).
\end{equation}
Then it is easy to show that
\begin{equation}\label{eq:psi1}
\psi_1(z;z) = \hbar \frac{d}{dx} \psi(z).
\end{equation}
(see Lemma 5.10 in \cite{BE}.) It follows that we can rewrite \eqref{eq:xi} as:
\begin{thm}\label{t:qc}
\begin{multline}
\left[\hbar^2 \frac{d^2}{dx^2} - 4 x(z)^3 + g_2(\tau) x(z) + g_3(\tau) 
-2 \hbar P_1(z;\tau)  \right. \\ \left. +2 \sum_{2g-2+n \geq 0} \frac{\hbar^{2g+n}}{n!} \left(\frac{\hat{G}_{g,n+1}(z';z)}{dz'} \right)_{z'=0}  \right] \psi(z) = 0.
\end{multline}
\end{thm}

\begin{proof}
We start with \eqref{eq:xi}, multiply by $\psi(z)$, to get
\begin{multline}
\hbar \psi(z) \frac{d}{dx} \xi_1(z;z) + \psi_1(z;z) \xi_1(z;z) + (- 4 x(z)^3 + g_2(\tau) x(z) + g_3(\tau))\psi(z) \\
-2 \hbar P_1(z;\tau) \psi(z) + 2 \psi(z) \sum_{2g-2+n \geq 0} \frac{\hbar^{2g+n}}{n!} \left(\frac{\hat{G}_{g,n+1}(z';z)}{dz'} \right)_{z'=0} = 0.
\end{multline}
But
\begin{align}
\hbar \psi(z) \frac{d}{dx} \xi_1(z;z) =&\hbar \frac{d}{dx} \psi_1(z;z) - \xi_1(z;z) \hbar \frac{d}{dx} \psi(z) \nonumber\\
=& \hbar^2 \frac{d^2}{dx^2} \psi(z) - \xi_1(z;z) \psi_1(z;z),
\end{align}
thus the equation becomes
\begin{multline}\label{eq:qcurve}
\hbar^2 \frac{d^2}{dx^2} \psi(z) + (- 4 x(z)^3 + g_2(\tau) x(z) + g_3(\tau))\psi(z) \\
-2 \hbar P_1(z;\tau) \psi(z) + 2 \psi(z) \sum_{2g-2+n \geq 0} \frac{\hbar^{2g+n}}{n!} \left(\frac{\hat{G}_{g,n+1}(z';z)}{dz'} \right)_{z'=0} = 0.
\end{multline}
\end{proof}

Theorem \ref{t:qc} gives an order two differential operator that annihilates the perturbative wave-function $\psi(z)$. However, this is not a quantum curve, according to Definition \ref{d:qc}. It has an infinite series of $\hbar$ corrections, and those corrections are not polynomials in $x$; in fact they are not even functions of $x$. They also have poles at the ramification points of the branched covering $\pi$.

What we have constructed is an order two differential operator that kills the standard perturbative wave-function \eqref{e:psi}, but it is not a quantum curve. However, it may be possible to define a new wave-function, which can be obtained from $\psi$, and that is annihilated by a proper quantum curve. To achieve this, we need to define the non-perturbative wave-function, which we will do in section \ref{s:NP}.

\begin{rem}
We remark here that we checked numerically on Mathematica that Theorem \ref{t:qc} is indeed satisfied for the first few orders in $\hbar$.
\end{rem}

\section{Perturbative wave-function: second approach}

\label{s:two}

In this section we study a second approach to obtain the order two differential operator that kills the perturbative wave-function. We will see that we obtain a differential operator that looks quite different \emph{a priori} from the differential operator obtained in the previous section. But we can prove that the two are equivalent. In fact, this equivalent gives rise to an infinite tower of identities for cycle integrals of elliptic functions.

We start with the topological recursion (Equation \ref{eq:TR}):
\begin{equation}
W_{g,n+1}(z_0, \mathbf{z}) = \text{d}z_{0}  \sum_{a \in R} \underset{z = a}{\text{Res}}~\left(\int_{\alpha}^{z} P_{2}(z' - z_{0}; \tau) \text{d}z'\right) \frac{\mathcal{R}^{(2)} W_{g,n+1}(z,-z; \mathbf{z})}{2\wp'(z;\tau)^{2} \text{d}z}.
\label{eq:TR2}
\end{equation}
We now wish to express the sum over residues around poles in $R$ in terms of residues of the other poles of the integrand. If the integrand was a well defined meromorphic differential in $z$ over the compact Riemann surface $\Sigma$, then the sum of its residues at all poles would have to vanish. However, it is not a well defined meromorphic differential in $z$; because of the line integral from $\alpha$ to $z$, it is only defined in the fundamental domain. Thus what we need to use is the Riemann bilinear identity.

\subsection{Riemann bilinear identity}
The integral form of the Riemann bilinear identity can be stated as follows:

\begin{equation*}
\sum_{\text{all poles }b \text{ of }u \eta}  \underset{z = b}{\text{Res}}~ u \eta = \frac{1}{2\pi i} \sum_{j = 1}^{g} \left( \oint_{B_{j}}\omega\oint_{A_{j}}\eta - \oint_{B_{j}}\eta\oint_{A_{j}}\omega \right),
\end{equation*}
where $\eta$ is a meromorphic differential on the compact Riemann surface $\Sigma$ of genus $g$, and $(A_j, B_j)$, $j=1,\ldots,g$ is a symplectic basis of cycles. Moreover,
\begin{equation}
u(z) = \int_{\alpha}^z \omega,
\end{equation}
where $\omega$ is a residueless meromorphic differential, $\alpha$ is an arbitrary base point, and the line integral is taken in the fundamental domain.

We can apply the Riemann bilinear identity to \eqref{eq:TR2}. First, we have that
\begin{multline}
 \sum_{\text{all poles }b} \underset{z = b}{\text{Res}}~\left(\int_{\alpha}^{z} P_{2}(z' - z_{0}; \tau) \text{d}z'\right) \frac{\mathcal{R}^{(2)} W_{g,n+1}(z,-z; \mathbf{z})}{2\wp'(z;\tau)^{2} \text{d}z} \\
 = \frac{1}{2 \pi i} \left(\oint_B  P_{2}(z -  z_{0}; \tau)\text{ d}z \oint_A \frac{\mathcal{R}^{(2)} W_{g,n+1}(z,-z; \mathbf{z})}{2\wp'(z;\tau)^{2} \text{d}z} \right. \\ \left. - \oint_A P_{2}(z -  z_{0}; \tau)\text{ d}z \oint_B \frac{\mathcal{R}^{(2)} W_{g,n+1}(z,-z; \mathbf{z})}{2\wp'(z;\tau)^{2} \text{d}z} \right).
\end{multline}
We note that:
\begin{equation}
\oint_{A} P_{2}(z -  z_{0}; \tau)\text{ d}z = 0, \qquad \oint_{B} P_{2}(z - z_{0}; \tau)\text{ d}z = 2 \pi i,
\end{equation}
thus
\begin{equation}
\sum_{\text{all poles }b} \underset{z = b}{\text{Res}}~\left(\int_{\alpha}^{z} P_{2}(z' - z_{0}; \tau) \text{d}z'\right) \frac{\mathcal{R}^{(2)} W_{g,n+1}(z,-z; \mathbf{z})}{2\wp'(z;\tau)^{2} \text{d}z} \\
 = B_{g,n+1}(\mathbf{z}) .
\end{equation}
where we defined
\begin{equation*}
B_{g,n+1} (\mathbf{z}):=  \oint_{A} \frac{\mathcal{R}^{(2)} W_{g,n+1}(z,-z; \mathbf{z})}{2\wp'(z;\tau)^{2} \text{d}z}.
\end{equation*}
But the poles $b$ can be separated into poles in $R$ and poles that are not in $R$, which means that we can rewrite \ref{eq:TR2} as:
\begin{multline}
\frac{W_{g,n+1}(z_0, \mathbf{z}) }{d z_0}=  B_{g,n+1}(\mathbf{z})\\ -  \sum_{a \notin R} \underset{z = a}{\text{Res}}~\left( P_{1}(z - z_{0}; \tau) - P_{1}(\alpha-z_{0}; \tau)\right) \frac{\mathcal{R}^{(2)} W_{g,n+1}(z,-z; \mathbf{z})}{2\wp'(z;\tau)^{2} \text{d}z}
\label{eq:TRinvert}
\end{multline}

Given this equation we now seek to write it in a form (nearly) identical to equation \eqref{eq:mess}. To do this we must first calculate the residues.

\subsection{Calculating the residues} 
Now since, $P_{1}(z-z_{0}) \rightarrow -\frac{1}{z-z_{0}}$ as $z \rightarrow z_0$, we see that there is a simple pole at $z = z_{0}$. However there is also a collection of poles at each of the marked points $z_{j}$ (with $j = \{1,\cdots,n\}$) coming from the recursive structure. To see this more clearly, let us examine equation \eqref{eq:R2Wgn} more closely:

\begin{multline}
 \mathcal{R}^{(2)} W_{g,n+1}(z,-z; \mathbf{z}) =   W_{g-1,n+2}(z,-z,\mathbf{z}) \\ +\sum_{\text{stable}} W_{g_1, |I|+1}(z,I) W_{g_2, |J|+1}(-z,J) \\ \nonumber
  + \sum_{j=1}^{n} W_{0,2}(z,z_{j})W_{g,n}(-z,\mathbf{z}/z_{j}) + W_{0,2}(-z,z_{j})W_{g,n}(z,\mathbf{z}/z_{j}).
\end{multline}

The ``stable" sum term excludes the cases where either $(g_1,|I| )$ or $(g_{2},|J|)$ is equal to $(0,0)$ or $(0,1)$.
 From here we note that as $z \rightarrow \pm z_{j}$, $W_{0,2}(\pm z,z_{j}) \rightarrow \pm\frac{\text{d}z\text{d}z_{j}}{(z \mp z_{j})^{2}}$, hence there are second order poles at each of the points $\pm z_{j}$ (with $j = \{1,\cdots,n\}$). Now we can proceed to calculate the residues:

\begin{multline}
\text{Residue at } z_0 = -\frac{\text{d}z_0}{2\wp'(z_0; \tau)^{2}}\left(\frac{W_{g-1,n+2}(z_{0},-z_{0},\mathbf{z})}{\text{d}z_0^{2}} \right. \\ \left.+ \sum_{g_1+g_2 =g} \sum_{I \cup J = \mathbf{z}}' \frac{W_{g_1, |I|+1}(z_0,I)}{\text{d}z_0} \frac{W_{g_2, |J|+1}(-z_0,J)}{\text{d}z_0} \right)
\end{multline}

\begin{equation}
\text{Residue at } z_i =  
\text{d}_{z_{i}} \left( \frac{P_{1}(z_{i} - z_{0}; \tau) - P_{1}(\alpha - z_{0}; \tau)}{2\wp'(z_{i},\tau)^{2}\text{d}z_{i}} W_{g,n}(-z_{i},\mathbf{z}/z_{i}) \right)\text{d}z_{0}
\end{equation}

\begin{equation}
\text{Residue at } (-z_i) =  
\text{d}_{z_{i}} \left( \frac{P_{1}(z_{i} + z_{0}; \tau) + P_{1}(\alpha - z_{0} ; \tau)}{2\wp'(z_{i},\tau)^{2}\text{d}z_{i}} W_{g,n}(-z_{i},\mathbf{z}/z_{i}) \right)\text{d}z_{0}
\end{equation}

Summing all of these contributions, dividing both sides by $\text{d}z_{0}$ and rearranging the expression we arrive at: 
\begin{prop}\label{c:mess2}
For $2g-2+n \geq 0$,
\begin{multline}
- \frac{W_{g-1,n+2}(-z_{0},z_{0},\mathbf{z})}{dx(z_{0})^2} + \sum_{g_1+g_2 =g} \sum_{I \cup J = \mathbf{z}} \frac{W_{g_1, |I|+1}(-z_{0},I)}{dx(z_{0})} \frac{W_{g_2, |J|+1}(-z_{0},J)}{dx(z_{0})} \\-  \sum_{i=1}^n \left(\frac{ dx(z_i)}{(x(z_{0})-x(z_i))^2}\frac{W_{g,n}(-z_{0}, \mathbf{z} \setminus \{ z_i \})}{ dx(z_{0}) }\right. \\ \left.-d_{z_i} \left(\frac{1}{x(z_{0})-x(z_i)} \frac{W_{g,n}(-z_i, \mathbf{z} \setminus \{ z_i \} )}{dx(z_i)} \right)  \right) \\
+ \sum_{i=1}^{n} \text{d}_{z_{i}} \left(\frac{2P_{1}(z_{i}; \tau)}{\wp'(z_{i}; \tau)} \frac{W_{g,n}(-z_i, \mathbf{z} \setminus \{ z_i \} )}{dx(z_i)} \right)   - 2 B_{g,n+1}(\mathbf{z})= 0.
\label{eq:mess2}
\end{multline}
For $(g,n)=(0,1)$,
\begin{multline}
 2 \frac{W_{0,2}(-z_0,z_1)}{dx(z_0)} \frac{W_{0,1}(-z_0)}{dx(z_0)} - \frac{dx(z_1)}{(x(z_0)-x(z_1))^2} \frac{W_{0,1}(-z_0)}{dx(z_0)} \\+d_{z_1} \left(\frac{1}{(x(z_0) - x(z_1)) }\frac{W_{0,1}(-z_1)}{dx(z_1)}  \right)
 +\text{d}_{z_1} \left(\frac{2P_1(z_1; \tau)}{\wp'(z_1; \tau)}\frac{W_{0.1}(-z_1)}{\wp'(z_1); \tau} \right) = 0,
\end{multline}
while for $(g,n) = (0,0)$,
\begin{equation}
2 \frac{W_{0,1}(-z_0) W_{0,1}(-z_0)}{dx(z)^2} - 4 x(z)^3 + g_2(\tau) x(z) + g_3(\tau) = 0.
\end{equation}
\end{prop}

This is to compare with Corollary \ref{c:mess} obtained in the previous section.

\subsection{Differential operator}

From Proposition \ref{c:mess2} we want to obtain a differential operator that annihilates the perturbative wave-function. We follow the procedure outlined in the previous section. We arrive at the following differential equation:
\begin{thm}\label{t:qc2}
\begin{multline}
\left[\hbar^2 \frac{\text{d}^2}{\text{d}x^2} - 2\hbar^{2}\frac{P_1(z; \tau)}{\wp'(z; \tau)}\frac{\text{d}}{\text{d}x} - 4 x(z)^3 + g_2(\tau) x(z) + g_3(\tau)   \right. \\ \left. -2 \sum_{2g-2+n \geq 0} \frac{\hbar^{2g+n}}{n!} \left(\int_0^z \cdots \int_0^z B_{g,n+1}(\mathbf{z})\right) \right] \psi(z) = 0.
\end{multline}
\end{thm}
This is to be contrasted with the differential operator that was obtained in Theorem \ref{t:qc}. The perturbative wave-function $\psi$ is the same for both Theorems. It is then expected that the two differential equations should be equivalent, even though they  look quite different \emph{a priori}.

\subsection{Connection with the calculation of the previous section}

In the previous section, we calculated a differential equation satisfied by $\psi$; in this section we also computed a differential equation satisfied by $\psi$, which looks different \emph{a priori}. Let us now show that they are the same.

First, let us compare Proposition \ref{c:mess2} with Corollary \ref{c:mess} of the previous section. In particular, for equation \eqref{eq:mess2} to be equivalent to equation \eqref{eq:mess}, the following identity must hold:

\begin{cor}\label{c:identity}
For the Weierstrass spectral curve, for $2g-2+n \geq 0$,
\begin{equation}\label{eq:identity}
B_{g,n+1}(\mathbf{z}) = -\left(\frac{W_{g,n+1}(-z_{0}, \mathbf{z})}{dz_{0}} \right)_{z_{0}=0} + \sum_{i=1}^{n} \text{d}_{z_{i}} \left(\frac{P_{1}(z_{i}; \tau)}{\wp'(z_{i}; \tau)} \frac{W_{g,n}(-z_i, \mathbf{z} \setminus \{ z_i \} )}{dx(z_i)} \right).
\end{equation}
\end{cor}
This is a non-trivial identity between elliptic functions; in fact, it gives an infinite tower of expressions for cycle integrals of elliptic functions. We will come back to that in the next subsection. But we note here that we can also  prove Corollary \ref{c:identity} directly. 

We start with \eqref{eq:TRinvert}, which we rewrite as
\begin{multline}
 B_{g,n+1}(\mathbf{z}) = 
- \frac{W_{g,n+1}(-z_0, \mathbf{z}) }{d z_0} \\+ \sum_{a \notin R} \underset{z = a}{\text{Res}}~\left( P_{1}(z - z_{0}; \tau) - P_{1}(\alpha-z_{0}; \tau)\right) \frac{\mathcal{R}^{(2)} W_{g,n+1}(z,-z; \mathbf{z})}{2\wp'(z;\tau)^{2} \text{d}z}.
\end{multline}
We take the limit as $z_0 \to 0$. We obtain
\begin{multline}
 B_{g,n+1}(\mathbf{z}) = 
- \left( \frac{W_{g,n+1}(-z_0, \mathbf{z}) }{d z_0} \right)_{z_0 = 0} \\+  \sum_{a \notin R} \underset{z = a}{\text{Res}}~\left( P_{1}(z; \tau ) - P_{1}(\alpha; \tau)\right) \frac{\mathcal{R}^{(2)} W_{g,n+1}(z,-z; \mathbf{z})}{2\wp'(z;\tau)^{2} \text{d}z}.
\end{multline}
The only poles in the sum over $a \notin R$ are at $z = \pm z_i$, for $i=1,\ldots,n$. The residue at $z=z_i$ gives rise to a term of the form
\begin{equation}
d_{z_i} \left( \frac{ P_{1}(z_i ;\tau) - P_{1}(\alpha; \tau)}{2\wp'(z_i;\tau) } \frac{W_{g,n}(-z_i ,\mathbf{z} \setminus \{ z_i \} )}{d x(z_i)} \right),
\end{equation}
while the residue at $z = - z_i$ gives rise to a term of the form
\begin{equation}
d_{z_i} \left( \frac{ P_{1}(z_i ; \tau) + P_{1}(\alpha; \tau)}{2\wp'(z_i;\tau) } \frac{W_{g,n}(-z_i ,\mathbf{z} \setminus \{ z_i \} )}{d x(z_i)} \right).
\end{equation}
Putting those together, we obtain
\begin{equation}
B_{g,n+1}(\mathbf{z}) = 
- \left( \frac{W_{g,n+1}(-z_0, \mathbf{z}) }{d z_0} \right)_{z_0 = 0}+  \sum_{i=1}^n d_{z_i} \left( \frac{ P_{1}(z_i ;\tau) }{\wp'(z_i;\tau) } \frac{W_{g,n}(-z_i ,\mathbf{z} \setminus \{ z_i \} )}{d x(z_i)} \right),
\end{equation}
which is precisely \eqref{eq:identity}.

We can also compare Theorem \ref{t:qc2} with Theorem \ref{t:qc} of the previous section. These give two differential equations satisfied by the perturbative wave-function $\psi$. 
For Theorems \ref{t:qc} and \ref{t:qc2} to be equivalent, we need the following equation for the wave-function $\psi$ to be satisfied:
\begin{cor}
\begin{multline}
\left[-\hbar P_1(z;\tau)  +\sum_{2g-2+n \geq 0} \frac{\hbar^{2g+n}}{n!} \left(\frac{\hat{G}_{g,n+1}(z';z)}{dz'} \right)_{z'=0} \right] \psi \\= - \left[\sum_{2g-2+n \geq 0} \frac{\hbar^{2g+n}}{n!} \left(\int_0^z \cdots \int_0^z B_{g,n+1}(\mathbf{z})\right) + \hbar^{2}\frac{P_1(z;\tau)}{\wp'(z;\tau)}\frac{\text{d}}{\text{d}x} \right] \psi.
\label{eq:tbp}
\end{multline}
\end{cor}
It turns out that we can indeed show directly that this is the case. We start with \eqref{eq:identity}, which is valid for $2g-2+n \geq 0$:
\begin{equation}
B_{g,n+1}(\mathbf{z}) = 
- \left( \frac{W_{g,n+1}(-z', \mathbf{z}) }{d z'} \right)_{z' = 0}+  \sum_{i=1}^n d_{z_i} \left( \frac{ P_{1}(z_i ;\tau) }{\wp'(z_i;\tau) } \frac{W_{g,n}(-z_i ,\mathbf{z} \setminus \{ z_i \} )}{d x(z_i)} \right).
\end{equation}
We integrate in $z_1, \ldots, z_n$ from $0$ to $z$, multiply by $\frac{\hbar^{2g+n}}{n!}$, and sum over $g$ and $n$ from $2g-2+n \geq 0$. We get:
\begin{multline}
\sum_{2g-2+n \geq 0} \frac{\hbar^{2g+n}}{n!} \left(\int_0^z \cdots \int_0^z B_{g,n+1}(\mathbf{z})\right) =\\ -\sum_{2g-2+n \geq 0} \frac{\hbar^{2g+n}}{n!} \left(\frac{\hat{G}_{g,n+1}(z';z)}{dz'} \right)_{z'=0} 
+ \sum_{2g-2+n \geq 0}  \frac{\hbar^{2g+n}}{(n-1)!} \frac{P_1(z;\tau)}{\wp'(z;\tau)^2} \frac{\hat{G}_{g,n}(z;z)}{dz} .\label{eq:tbr}
\end{multline}
Then, using \eqref{eq:psi1}, we notice that
\begin{align}
\hbar \frac{d}{dx} \psi (z) =& \psi_1(z;z) \\ \nonumber
=& \left[ \wp'(z;\tau) - \frac{1}{\wp'(z;\tau)} \sum_{2g-2+n \geq -1} \frac{\hbar^{2g+n}}{n!} \frac{\hat{G}_{g,n+1}(z;z)}{dz} \right] \psi(z).
\end{align}
Redefining the index in the sum and multiplying the equation by $\hbar$, we get
\begin{align}
\hbar^2 \frac{d}{dx} \psi (z) = \left[ \hbar \wp'(z;\tau) - \frac{1}{ \wp'(z;\tau)} \sum_{2g-2+n \geq 0} \frac{\hbar^{2g+n}}{(n-1)!} \frac{\hat{G}_{g,n}(z;z)}{dz} \right] \psi(z).
\end{align}
Going back to \eqref{eq:tbr}, we multiply by $- \psi(z)$ and use the above to rewrite it as
\begin{multline}
- \left[  \sum_{2g-2+n \geq 0} \frac{\hbar^{2g+n}}{n!} \left(\int_0^z \cdots \int_0^z B_{g,n+1}(\mathbf{z})\right) +  \hbar^2 \frac{P_1(z;\tau)}{\wp'(z;\tau)} \frac{d}{dx} \right] \psi(z) \\ \nonumber
= \left[ -\hbar P_1(z;\tau) +  \sum_{2g-2+n \geq 0} \frac{\hbar^{2g+n}}{n!} \left(\frac{\hat{G}_{g,n+1}(z';z)}{dz'} \right)_{z'=0} \right] \psi(z),
\end{multline}
which is precisely \eqref{eq:tbp}.

\section{Identities for elliptic functions}

\label{s:identities}

In this section we go back to Corollary \ref{c:identity} and explore its consequences for elliptic functions. We see that \eqref{eq:identity} gives rise to an infinite sequence of identities for cycle integrals of elliptic functions. Let us have a look at these identities for the first few levels in $2g-2+n$.

Let us start at the first level, namely $2g-2+n = 0$. We start with $(g,n) = (1,0)$. In this case \eqref{eq:identity} becomes
\begin{equation}
B_{1,1} = - \left( \frac{W_{1,1}(-z_0)}{d z_0} \right)_{z_0=0}.
\end{equation}
By definition
\begin{equation}
B_{1,1} =  - \oint_A \frac{P_2(2 z;\tau)}{2 \wp'(z;\tau)^2} dz,
\end{equation}
while from Appendix A, after a few simplifications, we obtain
\begin{equation}
\left( \frac{W_{1,1}(-z_0)}{d z_0} \right)_{z_0=0} =- \frac{G_4(\tau) (G_2(\tau)^2 - 5 G_4(\tau) )}{60 (20 G_4(\tau)^3 - 49 G_6(\tau)^2 )}.
\end{equation}
The identity is then
\begin{cor}\label{c:ellint}
\begin{equation}
 \oint_A \frac{P_2(2 z;\tau)}{\wp'(z;\tau)^2} dz =  \frac{G_4(\tau) ( 5 G_4(\tau)-G_2(\tau)^2 )}{30 (20 G_4(\tau)^3 - 49 G_6(\tau)^2 )}.
\end{equation}
\end{cor}
In other words, we obtain an explicit expression for the A-cycle integral of the elliptic function $\frac{P_2(2 z;\tau)}{\wp'(z;\tau)^2}$ in terms of quasi-modular forms. To emphasize the non-triviality of this expression, we provide an independent proof of this Corollary in Appendix B directly from the theory of elliptic functions.

Let us now study the other case at level $2g-2+n=0$, namely $(g,n) = (0,3)$. In this case \eqref{eq:identity} becomes:
\begin{multline}
\frac{B_{0,3}(z_1,z_2)}{dz_1 dz_2} = - \left(\frac{W_{0,3}(-z_{0}, \mathbf{z})}{dz_{0} d z_1 dz_2} \right)_{z_{0}=0} -  \frac{d}{dz_1} \left(\frac{P_{1}(z_{1};\tau)}{\wp'(z_{1};\tau)^{2}} P_{2}(z_{1} + z_{2};\tau) \right) \\
- \frac{d}{dz_2} \left(\frac{P_{1}(z_{2};\tau)}{\wp'(z_{2})^{2}} P_{2}(z_{1} + z_{2};\tau) \right) .
\end{multline}
But by definition
\begin{multline}
\frac{B_{0,3}(z_1,z_2) }{dz_1 dz_2}= - \oint_A \frac{P_2(z-  z_1;\tau) P_2(z+ z_2;\tau)}{2 \wp'(z;\tau)^2} d z\\ - \oint_A \frac{P_2(z +z_1;\tau) P_2(z - z_2;\tau)}{2 \wp'(z;\tau)^2} d z.
\end{multline}
Moreover, the result for $W_{0,3}(z_0, z_1, z_2)$ in Appendix A reads
\begin{multline}
\left(\frac{W_{0,3}(-z_{0}, z_1, z_2)}{dz_{0} dz_1 dz_2} \right)_{z_{0}=0} \\= - \frac{12}{\Delta} \sum_{i=1}^3  (20 G_4(\tau) - e_i^2) (e_i + G_2(\tau) ) P_2(z_1-\omega_i;\tau) P_2(z_2-\omega_i;\tau).
\end{multline}
Therefore, the identity becomes
\begin{cor}
\begin{multline}
\oint_A \frac{P_2(z-  z_1;\tau) P_2(z+ z_2;\tau)}{2 \wp'(z;\tau)^2} d z + \oint_A \frac{P_2(z +z_1;\tau) P_2(z - z_2;\tau)}{2 \wp'(z;\tau)^2} d z \\= \frac{d}{dz_1} \left(\frac{P_{1}(z_{1};\tau)}{\wp'(z_{1};\tau)^{2}} P_{2}(z_{1} + z_{2};\tau) \right) + \frac{d}{dz_2} \left(\frac{P_{1}(z_{2};\tau)}{\wp'(z_{2};\tau)^{2}} P_{2}(z_{1} + z_{2};\tau) \right)  \\
-\frac{12}{\Delta}\sum_{i=1}^3  (20 G_4(\tau) - e_i^2) (e_i + G_2(\tau) ) P_2(z_1-\omega_i;\tau) P_2(z_2-\omega_i;\tau).
\end{multline}
\end{cor}

We can continue producing such identities at higher levels of $2g -2 + n \geq 0$. In particular, for all cases with $n=1$, we obtain identities relating cycle integrals of elliptic functions to quasi-modular forms. It would be interesting to study whether these identities are of interest from the point of view of elliptic functions and quasi-modular forms. For instance, they may be related to the cycle integrals studied in \cite{GM}.

\section{Non-perturbative wave-function and quantum curve}

\label{s:NP}

In each of the previous two sections, we obtained a differential operator that annihilates the perturbative wave-function. However, these differential operators were not quantum curves, according to Definition \ref{d:qc}. 

In this section we will switch gears and study the \textit{non-perturbative topological recursion} formalism described in \cite{BoE,BorE1, E2,EM1}. It is expected from matrix models that this non-perturbative wave-function should be annihilated by quantum curve. This is what we study in this section.

\subsection{Non-perturbative wave-function}

Let us now introduce a non-perturbative wave-function, along the lines of \cite{BoE,BorE1,E2,EM1}. 

One of the major motivations for the definition of the non-perturbative partition function in \cite{BoE,BorE1,E2,EM1} is to construct a $\tau$-function for an arbitrary algebraic curve. $\tau$-functions in classical integrable systems are functions that satisfy Hirota bilinear equations. The Hirota equations are also equivalent to a self-replication property of the kernel. Either of these statements imply that there exists a system of differential equations, which we can use to get the quantum curve.  In \cite{BorE1}, it is conjectured that this non-perturbative partition function is indeed a $\tau$-function.

\subsubsection{Notation}

Before we write down the expression for the non-perturbative partition function, we need to define some fundamental objects that we will use.

We define a Jacobi theta function (called $ \theta_{11} $ in some references):\footnote{We make a choice of characteristics here in defining our theta function; it may be interesting to study other choices of characteristics.}

\begin{align}
\theta(z|\tau) = \sum_{n \in \mathbb{Z}} e^{i \pi (n + 1/2)^2 \tau + 2 \pi i (z + 1/2)(n + 1/2)}
\end{align} where $ z \in \mathbb{C} $ and $ \tau \in \mathbb{H} $.
We also define:
\begin{align}
\zeta_{\hbar} (\tau) = \frac{1}{2\pi i \hbar}\left(\oint_{\mathcal{B}} y \mathrm{d}x - \tau \oint_{\mathcal{A}} y \mathrm{d}x\right)
\end{align}
and introduce the following notation:
\begin{align}
\theta(\tau) &= \theta (\zeta_\hbar (\tau) |\tau) \\ \nonumber
\theta_\bullet(z|\tau) &= \theta (\zeta_\hbar(\tau) + z | \tau)
\end{align}
Then we define the perturbative partition function:
\begin{align}
Z_{\text{pert}}(\tau) = \text{exp}\left( \frac{1}{\hbar^2} \sum_{k\geq 0} \hbar^k F_k(\tau)\right),
\end{align}
where the $ F_k $'s are the free energies of the spectral curve defined in \eqref{eq:fg} ($F_0$ and $F_1$ can be defined independently; we refer the reader to \cite{BorE1,EO,EO2} for more details) . However, this partition function is non-modular, which is not what we expect from a ``true" partition function coming from a quantum field theory. In order to construct the non-perturbative partition function (conjectured to be a $\tau$-function), we  multiply the perturbative partition function by certain combinations of theta functions and their derivatives, which exactly cancel out the non-modularity (proved in \cite{E2,EM1}). This also ensures that the non-perturbative partition function is background independent. 

\subsubsection{Non-perturbative partition function and wave-function}

The non-perturbative partition function introduced in \cite{E2,EM1} is defined as a functional on the spectral curve:
\begin{align} \label{eq:nppart}
	Z_{\text{NP}}(\tau) &= \text{exp}\left(\frac{1}{\hbar^2} \sum_{k \geq0} \hbar^k F_k(\tau)\right) \\ \nonumber
  \times &\left\{ \sum_{r\geq 0} \frac{1}{r!} \sum_{\substack{h_j,d_j \geq 0 \\ 2h_j +d_j - 2 >0 }} \hbar^{\sum  2h_j +d_j - 2 }  \prod_{j=1}^{r} \left(\frac{F_{h_j}^{(d_j)}(\tau)  }{(2\pi i)^{d_j} d_j !} \right) \nabla^{(\sum_j d_j)} \theta(\tau) \right\} 
\end{align}
where
\begin{align}
F^{(d)}_{h}(\tau) &= \frac{1}{n!}\frac{1}{(2\pi i)^d d!} \oint_{\mathcal{B}} \dots \oint_{\mathcal{B}} W_{h,d} (z_1, \ldots, z_d),
\end{align}
and $W_{h,d}(z_1,\ldots,z_d)$ are the meromorphic differentials constructed from the spectral curve. Here, $ \nabla \theta ( \tau ) = \left(\frac{\mathrm{d}}{\mathrm{d}z} \theta (z|\tau)\right)|_{z=0}$ .

From the non-perturbative partition function, one can define a non-perturbative wave-function, following \cite{BoE, BorE1}. What we will call non-perturbative wave-function in this paper, and denote by $\psi_{\text{NP}}$, is the $(1 \vert 1)$-kernel of \cite{BoE,BorE1}. It is defined as a ``Schlesinger" transformation of the non-perturbative partition function:
\begin{align}\label{eq:kernel}
\psi_{\text{NP}}(p_1,p_2)  = \frac{\mathcal{T}_\hbar[y \text{d}x \rightarrow y \text{d}x + \hbar \omega^{p_2-p_1}]}{\mathcal{T}_\hbar[y \text{d}x]},
\end{align}
where $\omega^{p_2-p_1}$ was defined in \eqref{e:omega} (we removed the $p$-dependence for clarity). In the following, for the Weierstrass spectral curve we will choose the base point $p_1 = 0$, and consider the wave-function as a function of $p_2=z$. So we write
\begin{equation}
\psi_{\text{NP}}(z) =  \frac{\mathcal{T}_\hbar[y \text{d}x \rightarrow y \text{d}x + \hbar \omega^{z-0}]}{\mathcal{T}_\hbar[y \text{d}x]}.
\end{equation}

\subsubsection{Graphical interpretation}

It turns out that $\psi_{\text{NP}}$ has a nice graphical interpretation in terms of connected graphs satisfying certain properties (see \cite{BoE} for more details). We define $S_k(z)$s, $k \geq 0$ as follows:
\begin{equation}\label{eq:expansion}
\psi_{\text{NP}}(z) =  \exp\left(\frac{1}{\hbar} \sum_{k \geq 0} \hbar^k S_k(z) \right).
\end{equation}
Then, we can write a general expression for the $S_k$s. For $k\geq 2$,
\begin{multline}\label{eq:sk}
S_{k}(z) = \sum_{\substack{h_j,n_j,d_j \geq 0 \\ \sum_{j \geq 0} (2h_j + n_j +d_j - 2) = k - 1 }}  \frac{1}{j!}\left( \prod_j  \hspace{.3cm} G^{h_{j},d_{j}}_{n_{j}}(z) \right)\\
\times \left( V^{(d_{1}, d_{2}, \cdots, d_{j})}_\bullet - \delta_{(\sum_j n_j) , 0} V^{(d_1)}V^{(d_2)}\cdots V^{(d_j)} \right),
\end{multline}
where we used the notation
\begin{align}
G^{h,d}_{n}(z) &= \frac{1}{n!}\frac{1}{(2\pi i)^{d} d!} \int_{0}^{z} \dots \int_{0}^{z} \oint_{\mathcal{B}} \dots \oint_{\mathcal{B}} W_{h,n+d}(z_1,\ldots,z_{n+d}) \\
V^{(d_{1}, \cdots, d_{j})}_\bullet &= \frac{\partial}{\partial z_{1}} \cdots \frac{\partial}{\partial z_{j}} \log \left[ E \left( \exp \sum z_{i} \nabla^{d_{i}}\right)|_{z=0}\right].
\end{align}
Here, $E$ is defined as $E(\nabla^{d_i}) = \frac{\nabla^{d_i} \theta_\bullet(z|\tau)}{\theta_\bullet(z|\tau)}$, with $ \nabla = \frac{\mathrm{d}}{\mathrm{d}z}$. The undotted  $ V^{(\cdots)} $s are given by the same expressions but in terms of undotted theta functions. We note as well that by connectedness we have
\begin{align}
V^{(d_1,d_2)}_\bullet &= V^{(d_1+d_2)}_\bullet - V^{(d_1)}_\bullet V^{(d_2)}_\bullet ,\\
V^{(d_1,d_2,d_3)}_\bullet &= V^{(d_1+d_2+d_3)}_\bullet - V^{(d_1+d_2)}_\bullet V^{(d_3)}_\bullet - V^{(d_2+d_3)}_\bullet V^{(d_1)}_\bullet -V^{(d_3+d_1)}_\bullet V^{(d_2)}_\bullet \\ \nonumber
&+ 2 V^{(d_1)}_\bullet V^{(d_2)}_\bullet V^{(d_3)}_\bullet,
\end{align}
and so on.

$S_0$ and $S_1$ are defined differently. For the Weierstrass spectral curve, they are simply given by
\begin{align}
S_{0} &= \int_{0}^{z} \wp'(z)^{2} \mathrm{d} z, \\
S_{1} &= -\int_{0}^{z} \frac{\wp''(z)}{2\wp'(z)} \mathrm{d} z.
\end{align}

Coming back to \eqref{eq:sk}, we can write down the graphical expansion explicitly. The expansion was written down in \cite{BoE}; we rewrite it here for reference.\footnote{Note that a few typos in the expressions of \cite{BoE} were corrected here.}
\begin{align}
S_2(z) &= G^{0,0}_3 (z) + G^{1,0}_1(z) + G^{0,1}_2 V^{(1)}_\bullet + 
G^{1, 1}_0(z) \left( V^{(1)}_\bullet - V^{(1)} \right)  \\ \nonumber
&+ G^{0,2}_1 (z) V^{(2)}_\bullet+ 
G^{0, 3}_0(z) (V^{(3)}_\bullet - V^{(3)}),
\end{align}
\begin{align}
S_3(z) &= G^{0,0}_4(z) + G^{1,0}_2(z) + G^{0,1}_3(z) V^{(1)}_\bullet + G^{1,1}_1(z) V^{(1)}_\bullet + G^{0,2}_2(z) V^{(2)}_\bullet \\ \nonumber
&+ G^{1,2}_0(z) (V^{(2)}_\bullet - V^{(2)} ) + G^{0,3}_1(z) V^{(3)}_\bullet + G^{0,4}_0(z) (V^{(4)}_\bullet - V^{(4)} ) \\ \nonumber
& + \frac{1}{2} (G^{0,1}_2(z))^2 V^{(1,1)}_\bullet + G^{0,1}_2(z) G^{1, 1}_0(z) V^{(1,1)}_\bullet   \\ \nonumber
&+ \frac{1}{2} (G^{1, 1}_0(z))^2 ( V^{(1,1)}_\bullet - (V^{(2)})^2)+ G^{0,1}_2(z) G^{0,2}_1(z) V^{(1,2)}_\bullet \\ \nonumber 
& + G^{0,1}_2(z) G^{0,3}_0(z) V^{(1,3)}_\bullet + G^{0,3}_0(z) G^{1,1}_0(z) (V^{(1,3)}_\bullet - V^{(1)}V^{(3)} ) \\ \nonumber 
&+ \frac{1}{2} G^{0,2}_1(z) V^{(2,2)}_\bullet + G^{0,2}_1(z) G^{0,3}_0(z) V^{(2,3)}_\bullet   \\ \nonumber
&+ G^{0,2}_1(z) G^{1,1}_0(z) V^{(1,2)}_\bullet + \frac{1}{2} (G^{0,3}_0(z))^2 (V^{(3,3)}_\bullet - (V^{(3)})^2),
\end{align}
\begin{align}
S_4(z) &=G^{0,0}_5(z)   
+ G^{1,0}_3(z) + G^{2,0}_1(z)  + 
G^{0,1}_4(z) V^{(1)}_\bullet + G^{1,1}_2(z) V^{(1)}_\bullet  \\ \nonumber 
&+ G^{1,2}_1(z) V^{(2)}_\bullet + G^{0,3}_2(z) V^{(3)}_\bullet + 
G^{1,3}_0(z) (-V^{(3)} + V^{(3)}_\bullet ) + G^{0,4}_1(z) V^{(4)}_\bullet  \\ \nonumber 
&+ G^{0,5}_0(z) (-V^{(5)} + V^{(5)}{_\bullet} ) + G^{0,1}_2(z) G^{0,1}_3(z) V^{(1,1)}_\bullet +
G^{0,1}_3(z) G^{1,1}_0(z) V^{(1,1)}_\bullet  \\ \nonumber 
&+ G^{0,1}_2(z) G^{1,1}_1(z) V^{(1,1)}_\bullet +G^{1,1}_0(z) G^{1,1}_1(z) V^{(1,1)}_\bullet + 
G^{0,1}_3(z) G^{0,2}_1(z) V^{(1,2)}_\bullet  \\ \nonumber 
&+ G^{0,1}_2(z) G^{0,2}_2(z) V^{(1,2)}_\bullet + 
G^{0,2}_2(z) G^{1,1}_0(z) V^{(1,2)}_\bullet + 
G^{0,2}_1(z) G^{1,1}_1(z) V^{(1,2)}_\bullet  \\ \nonumber 
&+ G^{0,1}_2(z) G^{1,2}_0(z) V^{(1,2)}_\bullet + 
G^{1,1}_0(z) G^{1,2}_0(z) (-V^{(1)} V^{(2)} + V^{(1,2)}_\bullet) \\ \nonumber 
&+ G^{0,1}_3(z) G^{0,3}_0(z) V^{(1,3)}_\bullet + 
G^{0,1}_2(z) G^{0,3}_1(z) V^{(1,3)}_\bullet + 
G^{0,3}_1(z) G^{1,1}_0(z) V^{(1,3)}_\bullet  \\ \nonumber 
&+ G^{0,3}_0(z) G^{1,1}_1(z) V^{(1,3)}_\bullet + 
G^{0,1}_2(z) G^{0,4}_0(z) V^{(1,4)}_\bullet + G^{0,2}_3(z) V^{(2)}_\bullet  \\ \nonumber 
&+ G^{0,4}_0(z) G^{1,1}_0(z) (-V^{(1)} V^{(4)} + V^{(1,4)}_\bullet) + 
G^{0,2}_1(z) G^{0,2}_2(z) V^{(2,2)}_\bullet  \\ \nonumber 
&+ G^{0,2}_2(z) G^{0,3}_0(z) V^{(2,3)}_\bullet + G^{0,2}_1(z) G^{0,3}_1(z) V^{(2,3)}_\bullet  + G^{2,1}_0(z) (-V^{(1)} + V^{(1)}_{\bullet} )  \\ \nonumber
&+ G^{0,2}_1(z) G^{0,4}_0(z) V^{(2,4)}_\bullet + 
G^{0,3}_0(z) G^{0,3}_1(z) V^{(3,3)}_\bullet + 
G^{0,2}_1(z) G^{1,2}_0(z) V^{(2,2)}_\bullet   \\ \nonumber 
&+ 
\frac{1}{6} G^{0,1}_2(z)^3 V^{(1,1,1)}_\bullet + 
\frac{1}{2} G^{0,1}_2(z)^2 G^{1,1}_0(z) V^{(1,1,1)}_\bullet\\ \nonumber  
&+ \frac{1}{6} G^{1,1}_0(z)^3 (-(V^{(1)})^3 + V^{(1,1,1)}_\bullet) + 
\frac{1}{2} G^{0,1}_2(z)^2 G^{0,2}_1(z) V^{(1,1,2)}_\bullet   \\ \nonumber 
& +G^{0,1}_2(z) G^{0,2}_1(z) G^{1,1}_0(z) V^{(1,1,2)}_\bullet + \frac{1}{2} G^{0,2}_1(z) G^{1,1}_0(z)^2 V^{(1,1,2)}_\bullet \\ \nonumber
&+ \frac{1}{2} G^{0,1}_2(z)^2 G^{0,3}_0(z) V^{(1,1,3)}_\bullet + 
G^{0,1}_2(z) G^{0,3}_0(z) G^{1,1}_0(z) V^{(1,1,3)}_\bullet  \\ \nonumber 
&+ \frac{1}{2} G^{0,3}_0(z) G^{1,1}_0(z)^2 (-(V^{(1)})^2 V^{(3)} + V^{(1,1,3)}_\bullet) + 
\frac{1}{2} G^{0,1}_2(z) G^{0,2}_1(z)^2 V^{(1,2,2)}_\bullet   \\ \nonumber 
&+ \frac{1}{2} G^{0,2}_1(z)^2 G^{1,1}_0(z) V^{(1,2,2)}_\bullet + G^{0,1}_2(z) G^{0,2}_1(z) G^{0,3}_0(z) V^{(1,2,3)}_\bullet  \\ \nonumber  
 &+ G^{0,2}_1(z) G^{0,3}_0(z) G^{1,1}_0(z) V^{(1,2,3)}_\bullet + 
\frac{1}{2} G^{0,1}_2(z) G^{0,3}_0(z)^2 V^{(1,3,3)}_\bullet  \\ \nonumber 
&+ \frac{1}{2} G^{0,3}_0(z)^2 G^{1,1}_0(z) (-V^{(1)} (V^{(3)})^2 + V^{(1,3,3)}_\bullet) + 
\frac{1}{6} G^{0,2}_1(z)^3 V^{(2,2,2)}_\bullet \\ \nonumber 
&+ \frac{1}{2} G^{0,2}_1(z)^2 G^{0,3}_0(z) V^{(2,2,3)}_\bullet + \frac{1}{2} G^{0,2}_1(z) G^{0,3}_0(z)^2 V^{(2,3,3)}_\bullet \\ \nonumber 
&+ \frac{1}{6} G^{0,3}_0(z)^3 (-(V^{(3)})^3 + V^{(3,3,3)}_\bullet) + 
G^{0,3}_0(z) G^{1,2}_0(z) (-V^{(2)} V^{(3)} + V^{(2,3)}_\bullet) \\ \nonumber
&+ G^{0,3}_0(z) G^{0,4}_0(z) (-V^{(3)} V^{(4)} + V^{(3,4)}_\bullet)  + 
\frac{1}{2} G^{0,1}_2(z) G^{1,1}_0(z)^2 V^{(1,1,1)}_\bullet 
\end{align}

Let us remark here that the theta functions vanish at argument 0, and hence the undotted $ V^{(d)} $'s are not defined. However, they only appear along with terms $ G^{h,d}_n $  with $  n = 0 $  (which are constant in z). These terms only change the wave-function $ \psi $ by an overall (possibly $ \hbar $-dependent) constant, and hence do not pose an issue in the following.

\subsection{Quantization condition}

In the previous subsection, we defined a non-perturbative wave-function $\psi_{\text{NP}}$. We also studied its graphical interpretation in the form
\begin{equation}
\psi_{\text{NP}}(z) =  \exp\left(\frac{1}{\hbar} \sum_{k \geq 0} \hbar^k S_k(z) \right).
\end{equation}
This was considered as a formal asymptotic series in $\hbar$. But in general, the $S_k(z)$ will also depend on $\hbar$; they however will not have power series expansions in $\hbar$. 
Thus, for this expansion to be useful for us, the $S_k(z)$ should be independent of $\hbar$. This can be referred to as a \emph{quantization condition} for the spectral curve.\footnote{The quantization condition for spectral curves was explored in \cite{Gu,GS} and subsequently in \cite{BoE} for spectral curves in $\mathbb{C}^* \times \mathbb{C}^*$, in the context of knot theory and the AJ conjecture. In this context, the quantization condition has a beautiful interpretation in terms of algebraic K-theory. See \cite{GS} and also \cite{BoE} for more details.}

The problem comes from
\begin{align}
\zeta_{\hbar}(\tau) = \frac{1}{2\pi i \hbar}\left(\oint_{\mathcal{A}} y \mathrm{d}x - \tau \oint_{\mathcal{B}} y \mathrm{d}x\right),
\end{align}
which enters into the definition of $\theta_\bullet(z|\tau)$. In general, it depends on $\hbar$. The simplest way to ensure that the $S_k$'s do not depend on $\hbar$ is to set $\zeta_\hbar = 0 $.\footnote{We note here that this is not the only way however; see \cite{BoE,GS} for more details. It would be interesting to investigate more general elliptic spectral curves that satisfy the quantization condition.} Let us work out what this means for the Weierstrass spectral curve.
First,
\begin{align}
\oint \limits_{\mathcal{B}} y \mathrm{d}x &= \oint \limits_{\mathcal{B}} \wp'(z;\tau)^{2} \mathrm{d} z \\ \nonumber
		&= - \int\limits_{0}^{\tau} \wp''(z;\tau) \wp(z;\tau) \mathrm{d} z \\ \nonumber
		&= - \int\limits_{0}^{ \tau} (6\wp(z;\tau)^{2} - \frac{g_2(\tau)}{2}) \wp(z;\tau) \mathrm{d} z \\ \nonumber
		&= 	\left(-\frac{3}{5} g_{3}(\tau) z + \frac{2}{5} g_{2}(\tau) \zeta(z;\tau) \right)_{0}^{ \tau} \\ \nonumber
		&= -\frac{3}{5} g_{3}(\tau) \tau + \frac{2}{5} g_{2}(\tau) (2\pi i +  \tau G_{2}(\tau)).
\end{align}
Second,
\begin{align}
\oint \limits_{\mathcal{A}}  y \mathrm{d}x &= \left(-\frac{3}{5} g_{3}(\tau) z + \frac{2}{5} g_{2} (\tau)\zeta(z;\tau) \right)_{0}^{1} \\
		&= -\frac{3}{5} g_{3}(\tau)  + \frac{2}{5} g_{2}(\tau) G_{2} (\tau).
\end{align}
Therefore, we get that
\begin{equation}
\zeta_\hbar \equiv \frac{1}{\hbar} \left (\frac{2}{5} g_2(\tau) \right).
\end{equation}
This tells us that for $\zeta_\hbar = 0$, we should set $ g_2 (\tau)= 0 $, which fixes the isomorphism class of the elliptic curve (i.e., fixes $ \tau $ in the fundamental domain). Without loss of generality we can also choose $ g_3(\tau) = 4 $ to get the curve in the form $ y^2 = 4(x^3 -1) $. For this curve, the values of $ x = \wp(z;(\tau)) $ at the half-periods are the third roots of unity: $ 1, \omega , \omega^2  $. 

This elliptic curve is of course very special. It corresponds to the curve with $\tau = \exp\left( \frac{2 \pi i}{3} \right)$. Its $j$-invariant vanishes. It also has complex multiplication. In fact, after rescaling $y \to 2 y$, it becomes the curve 144A1 in Cremona's classification. It would be interesting to investigate what role these special properties of the elliptic curve play in the non-perturbative setting.

For the rest of this section we focus on that particular elliptic curve, hence we will remove the explicit $\tau$ dependence, since $\tau$ is now fixed. The formulae that we will derive for the non-perturbative $S_k$'s are only valid for this particular elliptic curve.

\subsection{Quantum curve}

The authors of \cite{BorE1} conjecture that the non-perturbative partition function \eqref{eq:nppart} is a $\tau$-function, \emph{i.e.} that it satisfies the Hirota equations. Assuming this conjecture, they argue that there should be a quantum curve $\hat{P}(\hat{x},\hat{y}; \hbar)$ of the spectral curve that kills the non-perturbative wave-function:
\begin{equation}
\hat{P}(\hat{x},\hat{y}; \hbar) \psi_{\text{NP}}(z) = 0.
\end{equation}
We refer the reader to \cite{BorE1} for the details of the argument. The goal of this subsection is to study whether this conjecture is true for the Weierstrass spectral curve $y^2 = 4 (x^3-1)$.

To be more precise, for the case of the Weierstrass spectral curve the conjecture can be formulated as follows:
\begin{conj} \label{qconj}
	 Consider the Weierstrass spectral curve $ P(x,y) = y^2 - 4 (x^3 -1) = 0 $. Then there is a unique quantum curve $ \hat{P}(\hat{x},\hat{y};\hbar) $  of the form
	 	\begin{align}\label{eq:npqcurve}
	\hat{P}(\hat{x},\hat{y};\hbar) = \hbar^2 \frac{\text{d}^2}{\text{d} x^2} - 4 (x^3-1)+ {\sum_{i \geq 1} \hbar^{2 i} A_{2 i} (x) \frac{\text{d}}{\text{d} x} } + {\sum_{j \geq 1} \hbar^{2 j} B_{2 j} (x)} ,
	\end{align}
	where the $ A_{i} (x) $ and $ B_{j} (x) $ are polynomials of x,
	  which kills the non-perturbative wave-function constructed by equation \eqref{eq:kernel}:
	 \begin{align}\label{eq:npqcurve1}
	 \hat{P}(\hat{x},\hat{y};\hbar) \psi_{\text{NP}}(z) = 0.
	 \end{align} 
	Note that only even powers of $\hbar$ appear in the quantum curve. Moreover, we conjecture that the $A_{2i}(x)$ are polynomials of degree at most $i-2$, and the $B_{2j}(x)$ are polynomials of degree at most $j$.
	\end{conj} 

The general form of a quantum curve for the spectral curve $y^2 = 4(x^3-1)$ was given in \eqref{eq:qcwe}. The fact that only even powers of $\hbar$ should appear is clear. It is easy to see that the $ S_k $'s transform as $ S_k(-z) = (-1)^{k+1} S_k(z) $ from the transformation properties of the $ W_{g,n} $ and the $ V_\bullet^{(\cdots)} $s. The WKB expansion of \eqref{eq:npqcurve1} is
	\begin{align}\label{eq:wkb}
	S_{k-1}'' + \sum_{l=0}^{k} S_l'S_{k-l}' + B_k + \sum_{l=0}^{k} S_{k-l}' A_{l+1} = 0.
	\end{align}
	As $ A_i(z) = A_i(-z) $ and $ B_j(z) = B_j(-z) $, we see that $  A_{2i+1} = B_{2i+1} = 0 $ for all $ i \in \mathbb{Z} $ .
	
	Uniqueness of the quantum curve also follows directly from the WKB expansion above, which uniquely defines the $A_{2i}(x)$ and $B_{2j}(x)$.
	
	As for the degree of $A_{2i}(x)$ and $B_{2j}(x)$, the conjectured bound is easily obtained by looking at the behaviour of the $S_k$'s at $z=0$ (\emph{i.e.}, the double pole of $x$). For all $k \geq 1$, $S_k$ is of order $3-k$ at $z=0$ (positive order meaning a zero of degree $3-k$, negative order meaning a pole of order $|3-k|$). As for $k=0$, $S_0'$ has a pole of order $3$. Now, \eqref{eq:wkb} (for a specified $k$) ensures that $ B_{2k} $ cannot have a pole of order greater than $2k$, while $A_{2k}$ cannot have a pole of order greater than $2k-4$, which justifies the bound stated in the conjecture, since $x$ has a double pole at $z=0$.

What remains to be proven however is that the non-perturbative wave-function is annihilated by a quantum curve at all. At the moment we do not have a complete proof of Conjecture \ref{qconj}. However, we computed correlation functions $W_{g,n}$ for the Weirstrass curve up to level $2g-2+n = 3$ --- see Appendix A. Focusing on the particular curve $y^2 = 4 (x^3-1)$, we then calculated the relevant cycle integrals to construct the $S_k$'s for the non-perturbative wave-function $\psi_{\text{NP}}$. Using Mathematica, we were then able to verify Conjecture \ref{qconj} to order $ \hbar^5 $:
	\begin{thm}
		The quantum curve to order $ \hbar^{5} $ is
		\begin{align}
		\hat{P}(\hat{x},\hat{y};\hbar) = \hbar^2 \frac{\text{d}^2}{\text{d} x^2} - 4(x^3-1)+ { \hbar^{4} \frac{1}{2^6 3^2} \frac{\text{d}}{\text{d} x} }  + {\hbar^{2} \frac{x}{2^2 3} + \hbar^{4} \frac{x^2}{2^8 3^3}  + O(\hbar^5) }
		\end{align}
		In particular, it satisfies the requirements of Conjecture \ref{qconj}.
	\end{thm} 
	
	\begin{proof}
	The proof is computational. Using the correlation functions calculate in Appendix A, restricting to the curve with $g_2=0$ and $g_3=4$, and evaluating the relevant cycle integrals on Mathematica, we obtain the following expressions for the non-perturbative $S_k$'s defined in the previous subsection:
		\begin{align}
	\nonumber S_0'(z) &= \wp'(z) \\ \nonumber
	S_1'(z) &= -\frac{3 \wp^2(z)}{\wp'(z)^2} \\ \nonumber
	S_2'(z) &= -\frac{\wp (z)}{24 \wp '(z)}-\frac{21 \wp (z)}{8 \wp '(z)^3}-\frac{45 \wp (z)}{2 \wp '(z)^5} \\ \nonumber
	S_3'(z) &= -\frac{1}{1152} -\frac{1}{24 \wp '(z)^2} -\frac{109}{16 \wp '(z)^4}-\frac{243}{2 \wp '(z)^6}-\frac{405}{\wp '(z)^8} \\ \nonumber
	S_4'(z) &= -\frac{\wp (z)^2}{13824 \wp '(z)}-\frac{\wp (z)^2}{1152 \wp '(z)^3}-\frac{31 \wp (z)^2}{64 \wp '(z)^5}-\frac{13641 \wp (z)^2}{128 \wp '(z)^7}-\frac{41769 \wp (z)^2}{16 \wp '(z)^9} \\ & \qquad-\frac{89505 \wp (z)^2}{8 \wp '(z)^{11}} \nonumber
	\end{align}
	Here primes refer to differentiation with respect to $x = \wp(z)$. Using those results it is straightforward to check that the non-perturbative wave-function $\psi_{\text{NP}}$ is annihilated by the differential operator above, up to order $\hbar^5$.
	\end{proof}

We note here that the $S_k'(z)$ that we calculated are rational functions of $\wp(z)$ and $\wp'(z)$, but we also remind the reader that these results are only valid for the particular Weierstrass spectral curve $y^2 = 4(x^3-1)$. For other Weierstrass curves, the non-perturbative $S_k$'s will generally depend on $\hbar$, and the expressions above will certainly not be valid. It is not clear to us whether we can reconstruct the WKB expansion of a quantization of the general Weierstrass spectral curve from the non-perturbative topological recursion.

\section{Conclusion and open questions}

\label{s:conclusion}

In this paper we studied how to quantize the Weierstrass spectral curve via the Eynard-Orantin topological recursion. More precisely, we investigated whether there exists a quantization of the Weierstrass spectral curve that kills the wave-function constructed from the topological recursion.

We first studied the naive question of whether there is such a quantization that kills the perturbative wave-function, as is the case for genus zero spectral curves. Not surprisingly, we did obtain a differential operator that annihilates the wave-function, but it is not a quantum curve according to Definition \ref{d:qc}. Nevertheless, we obtained this differential operator using two different approaches, and as a side result we produced an infinite tower of identities for cycle integrals of elliptic functions.

We then constructed a non-perturbative wave-function, which is a better candidate for a quantum curve, as expected from matrix models. By direct computations on Mathematica, we showed that indeed, up to order $\hbar^5$, the non-perturbative wave-function is killed by a non-trivial quantization of the Weierstrass spectral curve. However, we could only construct the non-perturbative wave-function if the quantization condition was satisfied; for this we focused on the simple elliptic curve $y^2 = 4(x^3-1)$.

There are many open questions that should be further studied. To name a few:
\begin{itemize}
\item Conjecture \ref{qconj} remains to be proven for the spectral curve $y^2 = 4(x^3-1)$. It is somehow expected to be true from the point of view of matrix models, but it would be very nice to have a formal proof of this conjecture. It may be possible to use our results from Sections \ref{s:one} and \ref{s:two} about the perturbative wave-function to construct a proof of the conjecture for the non-perturbative wave-function.
\item In the non-perturbative case, we restricted ourselves to the elliptic curve $y^2 = 4(x^3-1)$ to ensure that the non-perturbative wave-function has an expansion in $\hbar$. It would be very nice to see whether we can get rid of this constraint and study more general Weierstrass curves in the non-perturbative setting.
\item Via the non-perturbative approach, we obtained a proper quantization of the Weierstrass spectral curve. However, it is a rather non-trivial one, since it has an infinite number of $\hbar$ corrections. A more natural quantization would consist in 
\begin{equation}
\hat P(\hat x, \hat y) = \hbar^2 \frac{d^2}{d x^2} - 4 (x^3-1),
\end{equation}
that is, without $\hbar$ corrections. It would certainly be very interesting to study whether the WKB asymptotic solution to this equation can somehow be reconstructed, non-perturbatively, from the Eynard-Orantin topological recursion.
\item In this paper we focused on the Weierstrass spectral curve. But of course it would be very interesting to study larger classes of spectral curves, both at genus one and at higher genus, in the spirit of \cite{BE} for genus zero curves. The approach of \cite{BE} for the perturbative wave-function can certainly be generalized to higher genus curves, as we did in section \ref{s:one} (the general expresions become rather messy quickly though). But the most interesting question would be to study the non-perturbative wave-function. 
\item Coming back to the Weierstrass spectral curve, in Appendix A we calculated many correlation functions produced by the topological recursion. Generally speaking, in most applications of the topological recursion those correlation functions are generating functions for some interesting enumerative invariants. It is unclear at the moment whether there is such an interpretation for the correlation functions produced by the Weierstrass spectral curve. This is certainly worth investigating.
\item Finally, we obtained in section \ref{s:identities} an infinite sequence of identities for cycle integrals of elliptic functions. A natural question is whether those are interesting  from the point of view of elliptic functions and quasi-modular forms. In particular, they may be related to the results on cycle integrals of elliptic functions obtained in \cite{GM}. Moreover, the manipulations done in this paper are quite general, and would probably lead to analogous relations for higher genus curves, which would certainly be interesting to investigate. We hope to report on that in the near future.
\end{itemize}

\appendix

\section{Correlation functions for Weierstrass curve}
\label{s:correlation}

In this appendix we record the correlation functions constructed from the Eynard-Orantin topological recursion at the first few recursive levels. Those are needed to calculate the first few terms ($S_2$, $S_3$ and $S_4$) in the WKB expansion in section \ref{s:NP}.
 
First, at level $2g-2+n = 1$, we get:
\begin{multline}
W_{0,3}(z_0,z_1,z_2) \\ =   \frac{12}{\Delta}\mathrm{d} z_0 \mathrm{d} z_1 \mathrm{d} z_2\sum_{i=1}^3  (20 G_4(\tau) - e_i^2) P_2(z_0-\omega_i) P_2(z_1-\omega_i) P_2(z_2-\omega_i),
\end{multline}
and
\begin{multline}
W_{1,1}(z_0)  \\= \frac{6}{\Delta} \mathrm{d} z_0 \sum_{i=1}^3 (20 G_4(\tau) - e_i^2) \left( (G_2(\tau)-  e_i) P_2(z_0 - \omega_i,\tau) + \frac{1}{4!} P_2^{(2)}(z_0-\omega_i,\tau)  \right).
\end{multline}

At level $2g-2+n = 2$, we get

{\footnotesize
	\begin{multline*}
		W_{1,2}(z_{0},z_{1}) = \frac{1}{\Delta^2} \sum_{i=1}^{3} 9 (20 G_4-e_i^2) (-60 G_4 (e_i^2-20 G_4) \\
		(G_2+\wp (z_0-w_i)) (G_2+\wp (z_1-w_i))-(e_i^2-20 G_4) (6 \wp (z_0-w_i){}^3-30 G_4 \wp (z_0-w_i)\\
		+\wp '(z_0-w_i){}^2) (G_2+\wp (z_1-w_i))-6 (e_i^2-20 G_4) (5 G_4-\wp (z_0-w_i){}^2)\\
		 (-4 G_2^2+8 e_i G_2-\wp (z_1-w_i){}^2+5 G_4+(8 e_i-4 G_2) \wp (z_1-w_i))\\
		 +(G_2+\wp (z_0-w_i)) (-8 (e_{i+2}^2-20 G_4) (G_2+\wp (z_1-w_{i+2})) (e_{i+1}+G_2){}^2\\
		 -8 (e_i^2-20 G_4) (G_2^2+6 G_4) (G_2+\wp (z_1-w_i))\\
		 -60 (e_i^2+2 G_4) (e_i^2-20 G_4) (G_2+\wp (z_1-w_i))+2 ((4 G_2 (G_2-e_i)+G_4) (20 G_4-e_i^2)\\
		 +(20 G_4-e_{i+2}^2) (e_{i+1}^2+4 (e_{i+1}+G_2) (G_2-e_{i+2})-5 G_4)\\
		 +(20 G_4-e_{i+1}^2)  (e_{i+2}^2+4 (G_2-e_{i+1}) (e_{i+2}+G_2)-5 G_4)) (G_2+\wp (z_1-w_i))\\
		 -24 (e_i-G_2) (e_i^2-20 G_4) (5 G_4-\wp (z_1-w_i){}^2)\\
		 -8 e_i (e_i^2-20 G_4) (-12 G_2^2+4 e_i G_2-3 \wp (z_1-w_i){}^2+15 G_4\\
		 +4 (e_i-3 G_2) \wp (z_1-w_i))-8 (e_{i+2}+G_2){}^2 (e_{i+1}^2-20 G_4) (G_2+\wp (z_1-w_{i+1}))\\
		 +(e_i^2-20 G_4) (-6 \wp (z_1-w_i){}^3+30 G_4 \wp (z_1-w_i)-\wp '(z_1-w_i){}^2))),
		\end{multline*}}
\noindent and
{\footnotesize
\begin{multline*}
	W_{0,4}(z_{0},z_{1},z_{2},z_{3}) = \sum_{i=1}^{3} 3 \frac{144 (20 G_4-e_i^2)}{\Delta^2} (e_i^2-20 G_4)\\
	 (5 G_4-\wp (z_0-w_i){}^2) (G_2+\wp (z_1-w_i)) (G_2+\wp (z_2-w_i)) (G_2+\wp (z_3-w_i))\\
	 +\frac{144 (20 G_4-e_i^2)}{\Delta^2} (G_2+\wp (z_0-w_i)) (12 e_i (e_i^2-20 G_4) (G_2+\wp (z_1-w_i)) \\
	 (G_2+\wp (z_2-w_i)) (G_2+\wp (z_3-w_i))+3 (e_i^2-20 G_4) (5 G_4-\wp (z_1-w_i){}^2) \\
	 (G_2+\wp (z_2-w_i)) (G_2+\wp (z_3-w_i))+3 (e_i^2-20 G_4) (G_2+\wp (z_1-w_i))\\
	  (5 G_4-\wp (z_2-w_i){}^2) (G_2+\wp (z_3-w_i))+(-G_2 (e_i^2-20 G_4) (G_2+\wp (z_1-w_i)) \\
	  (G_2+\wp (z_2-w_i))-(e_{i+2}+G_2) (e_{i+1}^2-20 G_4) (G_2+\wp (z_1-w_{i+1}))\\
	   (G_2+\wp (z_2-w_{i+1}))  -(e_{i+1}+G_2) (e_{i+2}^2-20 G_4) (G_2+\wp (z_1-w_{i+2})) \\
	   (G_2+\wp (z_2-w_{i+2}))) (G_2+\wp (z_3-w_i))+3 (e_i^2-20 G_4) (G_2+\wp (z_1-w_i)) \\
	   (G_2+\wp (z_2-w_i)) (5 G_4-\wp (z_3-w_i){}^2)+(G_2+\wp (z_2-w_i)) (-G_2 (e_i^2-20 G_4) \\
	   (G_2+\wp (z_1-w_i)) (G_2+\wp (z_3-w_i))-(e_{i+2}+G_2) (e_{i+1}^2-20 G_4) \\
	   (G_2+\wp (z_1-w_{i+1})) (G_2+\wp (z_3-w_{i+1}))-(e_{i+1}+G_2) (e_{i+2}^2-20 G_4) \\
	   (G_2+\wp (z_1-w_{i+2})) (G_2+\wp (z_3-w_{i+2}))) +(G_2+\wp (z_1-w_i)) (-G_2 (e_i^2-20 G_4) \\
	   (G_2+\wp (z_2-w_i)) (G_2+\wp (z_3-w_i))-(e_{i+2}+G_2) (e_{i+1}^2-20 G_4) \\
	   (G_2+\wp (z_2-w_{i+1})) (G_2+\wp (z_3-w_{i+1}))-(e_{i+1}+G_2) (e_{i+2}^2-20 G_4) \\
	   (G_2+\wp (z_2-w_{i+2})) (G_2+\wp (z_3-w_{i+2})))),
\end{multline*}
}
\noindent where the index $i$ is defined mod $3$.

We also calculated the correlation functions at level $2g-2+n=3$, namely $W_{0,5}(z_0,z_1,z_2,z_3,z_4)$, $W_{1,3}(z_0,z_1,z_2)$ and $W_{2,1}(z_0)$. The expressions are very long though so we will not include them here. They are available upon request.

 \section{An independent proof of Corollary \ref{c:ellint}}
 
 \label{s:proof}

In this Appendix we provide an independent proof of Corollary \ref{c:ellint} directly from the theory of elliptic functions. Recall that Corollary \ref{c:ellint} states that:
\begin{equation}
 \oint_A \frac{P_2(2 z;\tau)}{\wp'(z;\tau)^2} dz =  \frac{G_4(\tau) ( 5 G_4(\tau)-G_2(\tau)^2 )}{30 (20 G_4(\tau)^3 - 49 G_6(\tau)^2 )}.
\end{equation}
Let us evaluate the period integral on the LHS explicitly and show that it is indeed equal to the quasi-modular form on the RHS. In this Appendix we will suppress the $\tau$-dependence for brevity.

First we expand the integrand with a ``double angle" identity:
\begin{equation}
P_2(2z) = G_2 - 2\wp(z) + \frac{1}{4} \left(\frac{\wp''(z)}{\wp'(z)}\right)^{2}
\end{equation}

Hence our original integral splits into the following three integrals:

\begin{equation}
\oint_A \frac{P_2(2 z)}{\wp'(z)^2} dz = G_2\oint_A \frac{\text{d}z}{\wp'(z)^2} - 2\oint_A\frac{\wp(z)}{\wp'(z)^2}\text{d}z + \frac{1}{4}\oint_A \frac{\wp''(z)^{2}}{\wp'(z)^{4}}\text{d}z
\end{equation}

Let us focus on the third constituent integral. Using integration by parts and the fact that $\wp'''(z) = 12\wp(z)\wp'(z)$ we see that it simplifies into a more familiar form:

\begin{equation}
\frac{1}{4}\oint_A \frac{\wp''(z)^{2}}{\wp'(z)^{4}}\text{d}z = \frac{1}{4}\left\{ \left. -\frac{\wp''(z)}{3\wp'(z)^{3}}\right|_0^1 + 4 \oint_A \frac{\wp(z)}{\wp'(z)^{2}}\right\} = \oint_A \frac{\wp(z)}{\wp'(z)^2}
\end{equation}

Hence our original problem reduces to solving only two integrals:

\begin{equation}\label{eq:P2integral}
\oint_A \frac{P_2(2 z)}{\wp'(z)^2} dz = G_2\oint_A \frac{\text{d}z}{\wp'(z)^2} - \oint_A\frac{\wp(z)}{\wp'(z)^2}\text{d}z 
\end{equation}

To solve both we need a very useful identity, which follows directly from the differential equation for the Weierstrass $\wp$-function \eqref{eq:WPdiffeq} and the fact that $\frac{2}{3}(\wp''(z)-g_2) = 4\wp(z)^2 - g_2$ :

\begin{equation}
\frac{1}{\wp'(z)^{2}}  = \frac{1}{g_3}\left[\frac{2}{3}\frac{\wp(z)\left(\wp''(z) - g_2\right)}{\wp'(z)^2} - 1 \right]
\end{equation}

As it turns out, using integration by parts we can express these two integrals in terms of one another:

\begin{align}\label{eq:1}
\oint_A \frac{\text{d}z}{\wp'(z)^{2}}  =& \frac{1}{g_3}\left[-1 + \frac{2}{3}\oint_A\frac{\wp(z)\left(\wp''(z) - g_2\right)}{\wp'(z)^2}\right] \\
=& -\frac{1}{3g_3}\left[1+ 2g_2\oint_A\frac{\wp(z)}{\wp'(z)^{2}}\text{d}z \right]
\end{align}

\begin{align}\label{eq:2}
\oint_A \frac{\wp(z)}{\wp'(z)^{2}}\text{d}z =& \frac{1}{g_3}\left[\frac{2}{3}\oint_A \frac{\wp(z)^2\left(\wp''(z) - g_2\right)}{\wp'(z)^2} - \oint_A \wp(z)\text{d}z \right] \\=& -\frac{1}{3g_3} \left[G_2 + \frac{g_2^2}{6} \oint_A \frac{\text{d}z}{\wp'(z)^2} \right] 
\end{align}
For the last equation, we used the fact that
\begin{equation}
\oint_A \wp(z)\text{d} z = -G_2,
\end{equation}
since
\begin{equation}
0 = \oint_A P_2(z) \text{d} z = \oint_A (\wp(z) + G_2) \text{d} z.
\end{equation}

Solving the system of equations \eqref{eq:1} and \eqref{eq:2} results in the following explicit expressions (with $\Delta = g_2^3 - 27g_3^2$):
\begin{equation}
\oint_A \frac{\text{d}z}{\wp'(z)^{2}}  = \frac{18 g_3 -12 G_2 g_2 }{2\Delta}
\end{equation}

\begin{equation}
\oint_A \frac{\wp(z)}{\wp'(z)^{2}}\text{d}z  = \frac{18 G_2 g_3 - g_2^2 }{2\Delta}
\end{equation}

As a result we see that the original integral \eqref{eq:P2integral} is given by:

\begin{equation}
\oint_A \frac{P_2(2z)}{\wp'(z)^2}\text{d}z = \frac{18 G_2 g_3 -12 G_2^2 g_2 }{2\Delta} - \frac{18 G_2 g_3 - g_2^2 }{2\Delta} = \frac{g_2(g_2-12G_2^2)}{2\Delta}
\end{equation}

Making the substitutions $g_2 = 60G_4$ and $g_3 = 140G_6$ we arrive at the final expected result:

\begin{equation}
\oint_A \frac{P_2(2z)}{\wp'(z)^2}\text{d}z = \frac{G_4(5G_4 - G_2^2)}{30(20G_4^3 - 49G_6^2)}
\end{equation}


\begin{thebibliography}{99}
\bibliographystyle{plain}


\bibitem{BeE}
M.~Berg\`ere and B.~Eynard, ``Determinantal formulae and loop equations," arXiv:0901.3273 [math-ph].


\bibitem{BeE2}
M.~Berg\`ere and B.~Eynard, ``Universal scaling limits of matrix models, and (p,q) Liouville gravity,'' arXiv.0909.0854 [math-ph].  

\bibitem{BoE}
G.~Borot and B.~Eynard, ``All-order asymptotics of hyperbolic knot invariants from non-perturbative topological recursion of A-polynomials,"  Quantum Topol. {\bf 6}, 39-138 (2015) [arXiv:1205.2261 [math-ph]].


\bibitem{BorE1}
G.~Borot and B.~Eynard, ``Geometry of spectral curves and all order integrable dispersive systems,'' SIGMA 8 (2012), 100, 53 p. [arXiv.1110.4936 [math-ph]].

\bibitem{BEMS}
G. Borot, B. Eynard, M. Mulase and B. Safnuk, ``A matrix model for simple Hurwitz numbers, and topological recursion," J. Geom. Phys. {\bf 61} (2), 522-540 (2011) [arXiv:0906.1206 [math-ph]].


\bibitem{BE}
V.~Bouchard and B.~Eynard, ``Reconstructing WKB from topological recursion," arXiv:1606.04498 [math-ph].

\bibitem{BE:2012}
V.~Bouchard and B.~Eynard,
``Think globally, compute locally,''
JHEP {\bf 02} (2013) 143 [arXiv:1211.2302v2 [math-ph]].


\bibitem{BHLMR}
V.~Bouchard, J.~Hutchinson, P.~Loliencar, M.~Meiers and M.~Rupert, ``A generalized topological recursion for arbitrary ramification,'' Annales Henri Poincar\'e, Vol.15 (2014), pp 143-169 [arXiv:1208.6035v1 [math-ph]].


\bibitem{BKMP}
V.~Bouchard, A.~Klemm, M.~Mari\~no and S.~Pasquetti, ``Remodeling the B-model," Commun.Math.Phys.287:117-178 (2009) [	arXiv:0709.1453 [hep-th]].

\bibitem{BM}
V.~Bouchard and M.~Mari\~no, ``Hurwitz numbers, matrix models and enumerative geometry," in \emph{From Hodge Theory to Integrability and tQFT: $tt*$-geometry}, Proceedings of Symposia in Pure Mathematics, AMS (2008) [arXiv:0709.1458 [math.AG]].

\bibitem{BSLM}
V.~Bouchard, D.~H.~Serrano, X.~Liu and M.~Mulase, ``Mirror symmetry for orbifold Hurwitz numbers," J. Differential Geom. Volume 98, Number 3 (2014), pp. 375-423 [arXiv:1301.4871 [math.AG]].


\bibitem{CEO}
L.~Chekhov, B.~Eynard and N.~Orantin, ``Free energy topological expansion for the 2-matrix model,''  JHEP {\bf 12} (2006), 053 [arXiv:math-ph/0603003].

\bibitem{DFM}
R.~Dijkgraaf, H.~Fuji, M.~Manabe,
"{The Volume Conjecture, Perturbative Knot Invariants, and Recursion Relations for Topological Strings}",
Nucl. Phys. {\bf B849} (2011), 166-211, arXiv: hep-th 1010.4542.

\bibitem{DLN}
N.~Do, O.~Leigh and P.~Norbury, ``Orbifold Hurwitz numbers and Eynard-Orantin invariants," arXiv:1212.6850 [math.AG].

\bibitem{DBOSS}
P. Dunin-Barkowski, N. Orantin, S. Shadrin and L. Spitz, ``Identification of the Givental formula with the spectral curve topological recursion procedure," Comm. Math. Phys. {\bf 328} 2, 669--700 (2014) [arXiv:1211.4021 [math-ph]].

\bibitem{DKOSS}
P. Dunin-Barkowski, M. Kazarian, N. Orantin, S. Shadrin and L. Spitz, ``Polynomiality of Hurwitz numbers, Bouchard-Mari\~no conjecture, and a new proof of the ELSV formula," Advances in Mathematics {\bf 279} 67--103 (2015) [arXiv:1307.4729 [math.AG]].

\bibitem{DBMNPS}
P. Dunin-Barkowski, M. Mulase, P. Norbury, A. Popolitov and S. Shadrin, ``Quantum spectral curve for the Gromov-Witten theory of the complex projective line," Journal f\"ur die reine und angewandte Mathematik (Crelles Journal), DOI: 10.1515/crelle-2014-0097 (2014) [arXiv:1312.5336 [math-ph]].

\bibitem{E1} B.~Eynard,
``Topological expansion for the 1-hermitian matrix model correlation functions,'' 
 JHEP {\bf 024} 0904 [arXiv:hep-th/0407261].
 
 \bibitem{E2} B.~Eynard,
 ``Large N expansion of convergent matrix integrals, holomorphic anomalies, and background independence,"
 JHEP {\bf 03} 003 (2009) [arXiv:0802.1788 [math-ph]]. 
 
 
 \bibitem{EM1}
B.~Eynard and M.~Mari\~no, ``A holomorphic and background independent partition function for matrix models and topological strings," J.Geom.Phys. {\bf 61}, 1181-1202 (2011) [arXiv:0810.4273 [hep-th]].

\bibitem{EMS}
 B.~Eynard, M.~Mulase and B.~Safnuk, ``The Laplace transform of the cut-and-join equation and the Bouchard-Marino conjecture on Hurwitz numbers," Publications of the Research Institute for Mathematical Sciences 47, 629--670 (2011) [arXiv:0907.5224 [math.AG]].
 
\bibitem{EO}
B.~Eynard and N.~Orantin,
``Invariants of algebraic curves and topological expansion,''
Comm.\ Numb.\ Theor.\ Phys.\ {\bf 1}, 347-452 (2007)
[arXiv:math-ph/0702045v4].

\bibitem{EO2}
B.~Eynard and N.~Orantin,
``Algebraic methods in random matrices and enumerative geometry,''
	arXiv:0811.3531v1 [math-ph].
	
	\bibitem{EO3}
B.~Eynard and N.~Orantin,
``Computation of open Gromov-Witten invariants for toric Calabi-Yau 3-folds by topological recursion, a proof of the BKMP conjecture," 	arXiv:1205.1103 [math-ph].

\bibitem{FLZ}
B.~Fang, C.-C.~M.~ Liu and Z.~Zong, ``All Genus Open-Closed Mirror Symmetry for Affine Toric Calabi-Yau 3-Orbifolds," arXiv:1310.4818 [math.AG].

\bibitem{FLZ2}
B.~Fang, C.~C.~M.~Liu and Z.~Zong, ``On the Remodeling Conjecture for Toric Calabi-Yau 3-Orbifolds," 	arXiv:1604.07123 [math.AG].

\bibitem{FLZ3}
B.~Fang, C.~C.~M.~Liu and Z.~Zong, ``The SYZ mirror symmetry and the BKMP remodeling conjecture," 	arXiv:1607.06935 [math.AG].
	
\bibitem{GM}
E.~Goujard and M.~Moeller,
``Counting Feynman-like graphs: Quasimodularity and Siegel-Veech weight,"
arXiv:1609.01658v1 [math.GT] .

\bibitem{GJKS}
J.~Gu, H.~Jockers, A.~Klemm and M.~Soroush, ``Knot Invariants from Topological Recursion on Augmentation Varieties," arXiv:1401.5095 [hep-th].

\bibitem{Gu}
S.~Gukov, ``Three-dimensional quantum gravity, Chern-Simons theory, and the A-polynomial,"
Comm. Math. Phys. {\bf 255} 577--627 (2005) [hep-th/0306165].

\bibitem{GS}
S.~Gukov and P.~Su\l kowski, ``A-polynomial, B-model, and Quantization,"
JHEP {\bf 1202} 070 (2012) [arXiv:1108.0002 [hep-th]].

\bibitem{Ma}
M.~Mari\~no, ``Open string amplitudes and large order behavior in topological string theory," JHEP {\bf 0803} 060 (2008) [arXiv:hep-th/0612127].

\bibitem{Mason:2010}
G.~Mason and M.~P.~Tuite,
``Vertex Operators and Modular Forms,''
in \emph{A Window into Zeta and Modular Physics}, ed. Kirsten, K. and Williams, F., MSRI Publications {\bf 57}, 183--278 CUP (2010) [arXiv:0909.4460v1 [math.QA]]. 

\bibitem{MSS}
M.~Mulase, S.~Shadrin and L.~Spitz, ``The spectral curve and the Schroedinger equation of double Hurwitz numbers and higher spin structures," Comm. Numb. Theor. Phys. {\bf 7} 1, 125--143 (2013) [arXiv:1301.5580 [math.AG].]



\end{thebibliography}
\end{document}